\newif\iffull
\fulltrue

\documentclass[sigconf]{acmart}
\iffull
\renewcommand\footnotetextcopyrightpermission[1]{} 
\fi

\AtBeginDocument{%
	\providecommand\BibTeX{{%
			\normalfont B\kern-0.5em{\scshape i\kern-0.25em b}\kern-0.8em\TeX}}}

\author{Panagiotis Chatzigiannis}
\affiliation{%
	\institution{George Mason University}
}
\email{pchatzig@gmu.edu}

\author{Foteini Baldimtsi}
\affiliation{%
	\institution{George Mason University}
}
\email{foteini@gmu.edu}

\author{Constantinos Kolias}
\affiliation{%
	\institution{University of Idaho}
}
\email{kolias@uidaho.edu}

\author{Angelos Stavrou}
\affiliation{%
	\institution{Virginia Tech}
}
\email{angelos@vt.edu}

\setcopyright{none} 
\usepackage{indentfirst}
\usepackage{xspace}
\usepackage{hyperref}

\usepackage{enumitem} 
\usepackage{pifont} 
\usepackage{subfig}
\usepackage{amsfonts}
\usepackage{amsthm} 

\usepackage{array} 
\usepackage{algorithm} 
\usepackage{algorithmic} 
\usepackage{pgf}
\usepackage{pgfplots}
\usepackage{graphicx}
\usepackage{siunitx}
\usepackage{tablefootnote}
\usetikzlibrary{arrows,chains,positioning,scopes,quotes}
\usepackage{wasysym}
\usepackage{multirow}
\usepackage[font=small,skip=4pt]{caption}
\usepackage{wrapfig}

\newcommand{\sysname}{\textsc{BBox-IoT}\xspace}
\newcommand{\privateSeed}{k_{n}}
\newcommand{\firstPublic}{k_{0}}

\newcommand{\signgen}[2]{({#1},{#2}) \leftarrow \mathsf{SignGen}(1^{\secpar})}
\newcommand{\sign}[3]{{#3} \leftarrow \mathsf{Sign}({#1},{#2}) }
\newcommand{\svrfy}[3]{\mathsf{SVrfy}({#1},{#2},{#3}):= b}

\newcommand{\otkeygen}[4]{({#1},{#2},{#3}) \leftarrow \mathsf{OTKeyGen}(1^{\secpar},{#4}) }
\newcommand{\otsign}[6]{({#6},{#1},{#3}) \leftarrow \mathsf{OTSign}({#2},{#5},{#4}) }
\newcommand{\otverify}[3]{\mathsf{OTVerify}({#1},{#2},{#3}):= b}

\newcommand{\sysparams}{\mathsf{pp}}
\newcommand{\publickey}[2]{\mathsf{pk_{{#1}_{#2}}}}
\newcommand{\secretkey}[2]{\mathsf{sk_{{#1}_{#2}}}}
\newcommand{\nstate}[1]{\mathsf{st_{#1}}}

\newcommand{\iotGroupSet}{\mathcal{G}}
\newcommand{\iotGroup}[1]{\mathsf{G_{#1}}}

\newcommand{\aggrset}[1]{\mathcal{AG}_{#1}}
\newcommand{\aggr}[2]{\mathsf{Ag_{#1 #2}}}
\newcommand{\publicagg}[2]{\publickey{A}{#1 #2}}
\newcommand{\secretagg}[2]{\secretkey{A}{#1 #2}}

\newcommand{\tx}{\mathsf{tx}}

\newcommand{\ordererset}{\mathcal{O}}
\newcommand{\orderer}[1]{\mathsf{O}_{#1}}
\newcommand{\node}{\mathsf{Node}}

\newcommand{\sensset}[1]{\mathcal{S}_{#1}}
\newcommand{\sens}[2]{\mathsf{S_{#1 #2}}}
\newcommand{\publicsens}[2]{\publickey{S}{#1 #2}}
\newcommand{\secretsens}[2]{\secretkey{S}{#1 #2}}

\newcommand{\sensorList}[1]{\mathsf{SL_{#1}}}

\newcommand{\trustedparty}[1]{\mathsf{TP_{#1}}}

\newcommand{\groupKey}[1]{\mathsf{K_{\iotGroup{#1}}}}

\newcommand{\localadmin}[1]{\mathsf{LAdm_{#1}}}

\newcommand{\blockchain}{\mathsf{BC}}
\newcommand{\block}[1]{\mathsf{B_{#1}}}

\newcommand{\tpsetup}{\mathsf{TPSetup}}
\newcommand{\aggsetup}{\mathsf{AggrSetup}}
\newcommand{\aggAdd}{\mathsf{AggrAdd}}
\newcommand{\senssenddata}{\mathsf{SensorSendData}}
\newcommand{\aggsendtx}{\mathsf{AggrSendTx}}
\newcommand{\aggagree}{\mathsf{AggrAgree}}
\newcommand{\consensus}{\mathsf{Consensus}} 

\newcommand{\updateBC}{\mathsf{UpdateBC}}

\newcommand{\configtxgen}{\mathsf{SystemInit}}
\newcommand{\msp}{\mathsf{MSP}}
\newcommand{\publicmsp}[1]{\publickey{MSP_{#1}}{}}
\newcommand{\secretmsp}[1]{\secretkey{MSP_{#1}}{}}

\newcommand{\ordererlist}{\mathsf{OL}}
\newcommand{\peerlist}{\mathsf{PL}}
\newcommand{\ladmlist}{\mathsf{LL}}

\newcommand{\agglistladmin}{\mathsf{AL}}
\newcommand{\senslistladmin}{\mathsf{SL}}
\newcommand{\senslistagg}{\mathsf{CL}}
\newcommand{\siglistagg}{\pset{}}
\newcommand{\txlist}{\txset{}}

\newcommand{\orderertxlist}{TXL}

\newcommand{\policy}{\mathsf{Pol}}
\newcommand{\mspopers}{\mathsf{oper}}
\newcommand{\updateconfig}{\mathsf{ConfigUpdate}}
\newcommand{\updatepolicy}{\mathsf{PolicyUpdate}}
\newcommand{\config}{\mathsf{Config}}
\newcommand{\orderersetup}{\mathsf{OrdererSetup}}
\newcommand{\ordereradd}{\mathsf{OrdererAdd}}
\newcommand{\localadmsetup}{\mathsf{LAdminSetup}}
\newcommand{\localadmreg}{\mathsf{LAdminAdd}}
\newcommand{\pset}[1]{\mathsf{pset_{#1}}}
\newcommand{\txset}[1]{\mathsf{txset_{#1}}}
\newcommand{\sensjoin}{\mathsf{SensorJoin}}
\newcommand{\aggupd}{\mathsf{AggrUpd}}

\newcommand{\readconfig}{\mathsf{ReadConfig}}

\newcommand{\noderem}{\mathsf{NodeRevoke}}
\newcommand{\grprem}{\mathsf{GroupRevoke}}
\newcommand{\aggremsens}{\mathsf{AggRevokeSensor}}

\newcommand{\cmark}{\ding{51}}%
\newcommand{\xmark}{\ding{55}}%
\newcommand{\adv}{\mathcal{A}}

\makeatletter
\newcommand{\thickhline}{%
	\noalign {\ifnum 0=`}\fi \hrule height 1pt
	\futurelet \reserved@a \@xhline
}
\newcolumntype{"}{@{\hskip\tabcolsep\vrule width 1pt\hskip\tabcolsep}}
\makeatother

\newcommand{\figurewidth}{0.48}

\newcommand{\pr}[1]   {\Pr\left[#1\right]}

\newcommand{\param}{\mathsf{pp}}

\newcommand{\secpar}{\lambda}

\newcommand{\Sign}{\mathsf{Sign}}

\def\comments_on{0} 

\ifnum\comments_on=1

\newcommand{\foteini}[1]{{
		{\color{red}\textbf{Foteini:}
			\emph{#1}
		}\normalcolor}}

\newcommand{\panagiotis}[1]{{
		{\color{blue}\textbf{Panagiotis:}
			#1
		}\normalcolor}}

\newcommand{\bingy}[1]{{
		{\color{magenta}\textbf{Kostas:}
			\emph{#1}
		}\normalcolor}}

\newcommand{\authnote}[1]{#1} 
\else
\newcommand{\bingsheng}[1]{}
\newcommand{\foteini}[1]{}
\newcommand{\aggelos}[1]{}
\newcommand{\panagiotis}[1]{}
\newcommand{\bingy}[1]{}
\newcommand{\authnote}[1]{}

\fi

\newtheorem{defn}{Definition}
\newtheorem{thm}{Theorem}

\newtheorem{corr}{Corollary}

\begin{document}
\title{Black-Box IoT: Authentication and Distributed Storage of IoT Data from Constrained Sensors}

\begin{abstract}

We propose Black-Box IoT (\sysname), a new ultra-lightweight black-box system for authenticating and storing IoT data.  
\sysname is tailored for deployment on IoT devices (including low-Size Weight and Power sensors) which are \emph{extremely constrained} in terms of computation, storage, and power. By utilizing core Blockchain principles, we ensure that the collected data is immutable and tamper-proof while preserving data provenance and non-repudiation. 
To realize \sysname, we designed and implemented a novel chain-based hash signature scheme which only requires hashing operations and removes all synchronicity dependencies between signer and verifier. Our approach enables low-SWaP devices to authenticate removing reliance on clock synchronization. Our evaluation results show that \sysname is practical in Industrial Internet of Things (IIoT) environments: even devices equipped with 16MHz micro-controllers and 2KB memory can broadcast their collected data without requiring heavy cryptographic operations or synchronicity assumptions. Finally,  when compared to industry standard ECDSA, our approach is two and three orders of magnitude faster for signing and verification operations respectively. Thus, we are able to increase the total number of signing operations by more than 5000\% for the same amount of power.
\end{abstract}



\maketitle
\iffull
\pagestyle{plain}
\fi

\section{Introduction}

The commercial success of low Size Weight and Power (SWaP) sensors and
IoT devices has given rise to new sensor-centric applications
transcending traditional industrial and closed-loop
systems~\cite{zou2018towards,derhamy2015survey}.  In their most recent
Annual Internet Report~\cite{cisco2020}, CISCO estimates that there
will be 30 billion networked devices by 2023, which is more than three
times the global population.  While very different in terms of their
hardware and software implementations, Industrial IoT (IIoT) systems
share common functional requirements: they are designed to collect
data from a large number of low-SWaP sensor nodes deployed at the
edge. These nodes, which we refer to as edge \emph{sensors}, are
resource-constrained devices used in volume to achieve a broader
sensing coverage while maintaining low cost. Thus, while capable of
performing simple operations, low-SWaP sensors usually depend on
battery power, are equipped with limited storage, and have low
processing speed~\cite{8534563}.

In practice, edge sensors are usually controlled by and report to more
powerful gateway devices (which we refer to as \emph{aggregators})
that process and aggregate the raw sensory data. For instance, in an
Industrial (IIoT) environment, sensors are devices such as temperature
sensors are broadcasting their measurements to the network router,
which in turn submits it to the cloud through the Internet. Until
recently, due to processing and storage constraints, many IoT designs
were geared towards direct to cloud aggregation and data
processing. However, latency, bandwidth, autonomy and data privacy
requirements for IoT applications keep pushing the aggregation and
processing of data towards the edge \cite{lin2019computation}. In
addition, in most use cases, IoT devices need to be mutually
\emph{authenticated} to maintain system integrity and the data origin
has to be verified to prevent data pollution
attacks~\cite{DBLP:journals/sj/LiuXG14,5638628} and in ``model
poisoning'' where an attacker has compromised a number of nodes acting
cooperatively, aiming to reduce the accuracy or even inject backdoors  to the resulting analysis
models~\cite{DBLP:conf/icml/BhagojiCMC19,gu2019badnets}.

The use of distributed, immutable ledgers has been proposed as a
prominent solution in the IoT setting allowing rapid detection of
inconsistencies in sensory data and network communications, providing
a conflict resolution mechanism without relying on a trusted
authority~\cite{DBLP:conf/icse/BellLBS17}.  A number of relevant schemes has been proposed in the
literature~\cite{DBLP:conf/lcn/ProfentzasAL19,Shafagh:2017:TBA:3140649.3140656},
which \iffull as we discuss in Section~\ref{relwork}, \fi propose
various ways to integrate distributed ledgers (commonly referred to as
\emph{Blockchain}) with IoT.

\noindent \textbf{The Challenge:} One of the main roadblocks for using
Blockchain-based systems as ``decentralized'' databases for sharing
and storing collected data is their dependency on asymmetric
authentication techniques. Typically, to produce authenticated data
packets, sensors have to digitally sign the data by performing public
key cryptographic operations, which are associated with expensive sign
and verification computations and large bandwidth
requirements. Although some high-end consumer sensor gateways and
integrated sensors might be capable of performing cryptographic
operations, a large number of edge sensors have limited computational
power, storage and energy~\cite{boyer2017distributed,KARRAY201889}.
To make matters worse, sensors try to optimize their power consumption
by entering a ``sleep'' state to save power resulting in intermittent
network connectivity and lack of synchronicity. Given such tight
constraints, an important challenge is allowing low-SWaP devices being
extremely constrained both in terms of computational power and memory
(categorized as Class 0 in RFC 7228~\cite{RFC7228}
ref. Section~\ref{measurements-scenario}), to authenticate and utilize
a blockchain-based data sharing infrastructure.

\noindent \textbf{Our Contributions:} We design and implement
\sysname, a complete blockchain-based system for Industrial IoT
devices aimed to create a decentralized, immutable ledger of sensing
data and operations while addressing the sensor and data
authentication challenge for extremely constrained devices. We aim to
use our system as a "black-box" that empowers operators of an
IIoT enclave to audit sensing data and operational information such as
IIoT communications across all IIoT devices.

To perform sensor and data authentication operations \emph{without}
relying on heavy cryptographic primitives, we introduce a novel
hash-based digital signature that uses an onetime hash chain of
signing keys. While our design is inspired by TESLA broadcast
authentication protocol~\cite{848446,Perrig02thetesla}, our approach
\emph{does not} require any timing and synchronicity assumptions
between signer and verifier. Overcoming the synchronicity
requirement is critical for low-SWaP devices since their internal
clocks often drift out of synchronization (especially those using low
cost computing
parts)~\cite{DBLP:journals/sensors/Tirado-AndresRA19,DBLP:journals/tii/ElstsFDOPC18}.
\iffull Our signature construction is proven secure assuming a
pre-image resistant hash function. Also, we achieved logarithmic
storage and computational (signature/verification) costs using
optimizations~\cite{jakobsson2002fractal,RSA:YSEL09}.  \fi Our
proposed scheme further benefits by the broadcast nature of the
wireless communication. Indeed, in combination with the immutable
blockchain ledger, we are able to ferret out man-in-the-middle attacks
in all scenarios where we have more than one aggregators in the
vicinity of the sensors. To bootstrap the authentication of sensor
keys, we assume an operator-initiated device bootstrap protocol that
can include either physical contact or wireless pairing using an
operator-verified ephemeral code between sensors and their receiving
aggregators. Our bootstrap assumptions are natural in the IoT setting,
where sensors often ``report" to specific aggregators and allows us to
overcome the requirement for a centralized PKI.  Note that our
signature scheme is of independent interest, in-line with recent
efforts by NIST for lightweight cryptography~\cite{turan2019status}.

For the blockchain implementation where a \emph{consensus} protocol is
needed, we consider a \emph{permissioned} setting, where a trusted
party authorizes system participation at the aggregator level.  Our
system supports two main types of IoT devices: low-SWaP sensors who
just broadcast data and self-reliant aggregators who collect the data
and serve as gateways between sensors and the blockchain.  While our
system is initialized by a trusted operator, the operator is not
always assumed present for data sharing and is only required for
high-level administrative operations including adding or removing
sensors from the enclave. We build the consensus algorithms for
\sysname using a modified version of Hyperledger Fabric
\cite{DBLP:journals/corr/abs-1801-10228}, a well known permissioned
blockchain framework, and leverage blockchain properties for
constructing our protocols tailored for constrained-device
authentication. However, \sysname operations are designed to be
lightweight and do not use public key cryptography based on the RSA or
discrete logarithm assumptions, which are common, basic building
blocks of popular blockchain implementations. We describe our system
in details considering interactions between all participants and argue
about its security.

We implemented and tested a \sysname prototype in an IIoT setting
comprising of extremely constrained sensors (Class 0 per RFC 7228). We
employed 8-bit sensor nodes with 16MHz micro controllers and 2KB RAM,
broadcast data every 10 seconds to a subset of aggregators (e.g. IIoT
gateways) which in turn submit aggregated data to a cloud
infrastructure. The evaluation shows that the IIoT sensors can compute
our 64-byte signature in 50ms, making our signature scheme practical
even for the least capable of IIoT platforms. Our evaluation section
shows results by considering a sensor/gateway ratio of 10:1. When
compared with ECDSA signing operations, our scheme is significantly
more efficient offering two (2) and three (3) orders of magnitude
speedup for signing and verification respectively. Our theoretical
analysis and implementation shows that we can achieve strong chained
signatures with half signature size, which permits accommodating more
operations in the same blockchain environment. \sysname is also over
50 times more energy-efficient, which makes our system ideal for edge
cost-efficient but energy-constrained IIoT devices and applications.

Finally, we adopt the same evaluation for Hyperledger Fabric considered in previous work~\cite{DBLP:journals/corr/abs-1801-10228} and estimate the end-to-end costs of \sysname when running on top of our Hyperledger modification, showing it is deployable in our considered use-cases.

\section{Background \& Preliminaries}
\label{backgrnd-prelims}
\subsection{Blockchain System Consensus}
\label{prel:consensus}
\iffull In distributed ledgers including Blockchains, we can
categorize the participants to: a) Blockchain maintainers also called
\emph{miners}, who are collectively responsible for continuously
appending valid data to the ledger, and b) clients, who are reading
the blockchain and posting proposals for new data. While clients are
only utilizing the blockchain in a read-only mode, the blockchain
maintainers who are responsible for ``book-keeping" must always act
according to a majority's agreement to prevent faulty (offline) or
Byzantine (malicious) behavior from affecting its normal
functionality. Reaching agreement requires the existence of a
\emph{consensus} protocol among these maintainers.

Moreover, the type of consensus protocol can be permissioned or
permissionless. In permissionless blockchains, like
Bitcoin~\cite{Nakamoto:bitcoin} and
Ethereum~\cite{buterin2014ethereum}, anyone can participate in the
consensus protocol and sybil attacks are prevented by new consensus
mechanisms such as Proof-of-Work or
Proof-of-Stake~\cite{C:KRDO17}. Although the open nature of
permissionless blockchains seems attractive, it does not really fit
the membership and access control requirements for IoT
deployments. Instead, operators would like to exert control over IoT
sensors and aggregators by authenticating all system participants.

Focusing on permissioned blockchains, consensus is computationally
easier to achieve even if nodes are of limited
capabilities. Distributed consensus mechanisms such as Practical
Byzantine Fault Tolerance (PBFT)~\cite{Castro:1999:PBF:296806.296824}
are relatively efficient, as long as the number of maintainers is
relatively small: Byzantine Fault Tolerant consensus is considered
scalable up to 20
nodes~\cite{DBLP:conf/ifip11-4/Vukolic15,DBLP:conf/middleware/ChondrosKR12}.
The distributed consesus mechanisms use inexpensive cryptographic
operations such as message authentication codes (MACs). Alternatively,
a trusted party $\trustedparty{}$ decides who should be given
participating access by issuing a private credential to each node.

In the following paragraphs we define consensus in the permissioned
setting and establish the properties needed for our system.
In Appendix \ref{apdx:consensus} we present an overview of  consensus algorithms in the permissioned setting 
inspired by
\cite{DBLP:journals/corr/abs-1711-03936,DBLP:journals/corr/CachinV17,RSA:GarKia20},
and discuss how each one would fit to \sysname.

\textbf{Fundamental consensus properties:} 
\label{prel:consensus:fundam-properties} Informally, the
consensus problem considers a set of state machines, we also refer to
them as \emph{replicas} or \emph{parties}, starting with the same
initial state. Then given a common sequence of messages, each correct
replica, by performing its private computation, should reach a state
consistent with other correct replicas in the system, despite any
possible network outages or other replica
failures~\cite{DBLP:journals/corr/abs-1711-03936}.

Since our system uses a blockchain, we focus on the notion of
\emph{ledger} consensus~\cite{RSA:GarKia20}, where a number of parties
receive a common sequence of messages and append their outputs on a
public ledger. As a special case, an \emph{Authenticated} ledger
consensus
protocol~\cite{Lamport:1982:BGP:357172.357176}\cite{EPRINT:LinLysRab04}
only permits participation through credentials issued by a
$\trustedparty{}$. The two fundamental properties of a ledger
consensus protocol are \cite{RSA:GarKia20}: (a) \emph{Consistency}: An
honest node's view of the ledger on some round $j$ is a prefix of an
honest node's view of the ledger on some round $j+\ell, \ell>0$.  (b)
\emph{Liveness}: An honest party on input of a value $x$, after a
certain number of rounds will output a view of the ledger that
includes $x$. We provide formal definitions in Appendix
\ref{apdx:consensus}.

\else In distributed ledgers (or Blockchains), we can categorize the
participants as follows: a) Blockchain maintainers (called also
\emph{miners}), who are collectively responsible for continuously
appending valid data to the ledger, and b) clients, who are reading
the blockchain and posting proposals for new data. While clients are
only utilizing the blockchain in a read-only mode, the blockchain
maintainers who are responsible for ``book-keeping" must always act
according to a majority's agreement to prevent faulty (offline) or
Byzantine (malicious) behavior from affecting its normal
functionality. This assumes that a \emph{consensus}
protocol takes place behind the scenes among these maintainers, which are distinguished to
permissioned or permissionless, according to their participation
controls.

Although the open nature of permissionless blockchains seems
attractive, it does not really fit the membership and access control
requirements for IoT deployments. In such settings, operators prefer to control
the participation of IoT sensors and aggregators by means of
authenticating them. Moreover, in permissioned settings consensus is computationally cheaper
and thus better suited to nodes with limited capabilities.

\textbf{Fundamental Consensus Properties:} 
\label{prel:consensus:fundam-properties} Informally, the ledger consensus problem \cite{RSA:GarKia20} considers a number of parties receiving a common sequence of messages, appending their outputs on a public ledger. The  two basic  properties of a ledger consensus protocol are:
(a) \emph{Consistency}: An honest node's view of the ledger on some round $j$ is a prefix of an honest node's view of the ledger on some round $j+\ell, \ell>0$.  
(b) \emph{Liveness}: An honest party  on input of value $x$, after a certain number of rounds outputs a ledger view that includes $x$. 
\fi

\textbf{\sysname Permissioned Consensus:}
The aforementioned fundamental properties are not sufficient for our
system consensus. For instance, most ``classical" consensus algorithms
such as PBFT~\cite{Castro:1999:PBF:296806.296824} have not been widely
deployed due to various practical issues including lack of
scalability. Taking the \sysname requirements into account, the system's
consensus algorithm needs to satisfy the following additional
properties:

\begin{enumerate}[leftmargin=*,label=\roman*.]
	\item Dynamic membership: In  \sysname, there is no a priori knowledge of system participants. New members might want to join (or leave) after bootstrapping the system. We  highlight that the vast majority of permissioned consensus protocols assume a static membership \iffull\cite{DBLP:conf/middleware/ChondrosKR12}\fi. Decoupling ``transaction signing participants" from ``consensus participants" is a paradigm that circumvents this limitation \iffull\cite{DBLP:journals/tdsc/RodriguesLCLS12}\fi.
	\item Scalable: \sysname might be deployed in wide-area scenarios (e.g. IIoT), so the whole system must support in practice many thousands of participants, and process many operations per second (more than 1000 op/s)\iffull~\cite{DBLP:conf/ifip11-4/Vukolic15}\fi. 
	\item DoS resistant: For the same reason above, participants involved in consensus should be resilient to denial-of-service attacks \iffull\cite{DBLP:journals/corr/abs-1711-03936}\fi.
\end{enumerate}
\iffull
In Appendix \ref{apdx:consensus-properties} we also mention some desirable (but not required) consensus properties, and show how some well-known consensus protocols satisfy our required properties above in Table \ref{prelims:consensus-comparison}.
\fi

\begin{figure*}
	\subfloat[Original architecture. Clients collect signatures from peers for a transaction, then submit the signed transaction to the ordering service which then returns a block containing packaged transactions to peers.]{
		\includegraphics[width=0.45\textwidth]{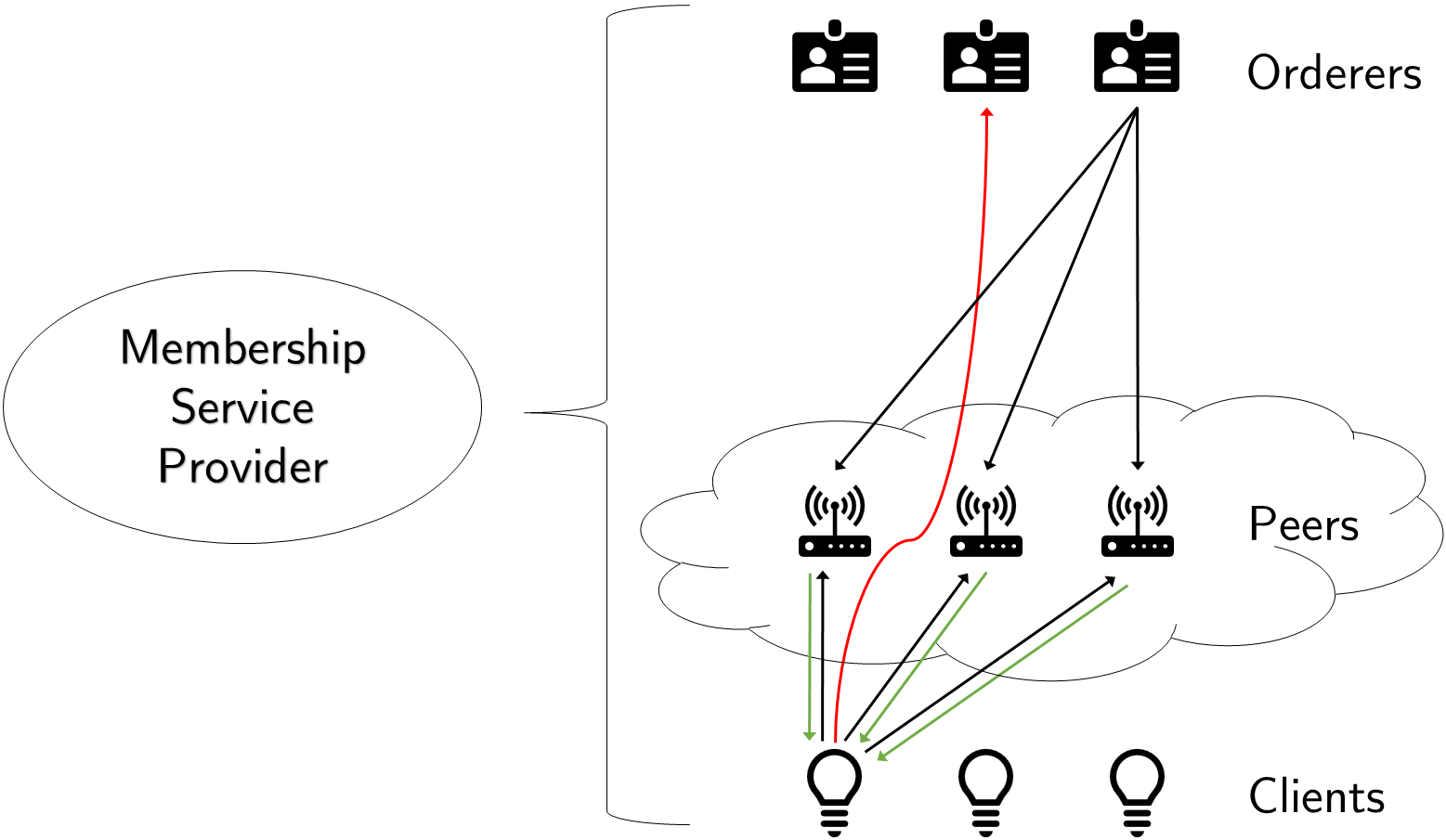} \label{HLF-fig}
	}
    \qquad
	\subfloat[Modified architecture. Clients only broadcast the transaction to the peers, who are then responsible themselves for signing it before submitting it to the ordering service.]{
		\includegraphics[width=0.45\textwidth]{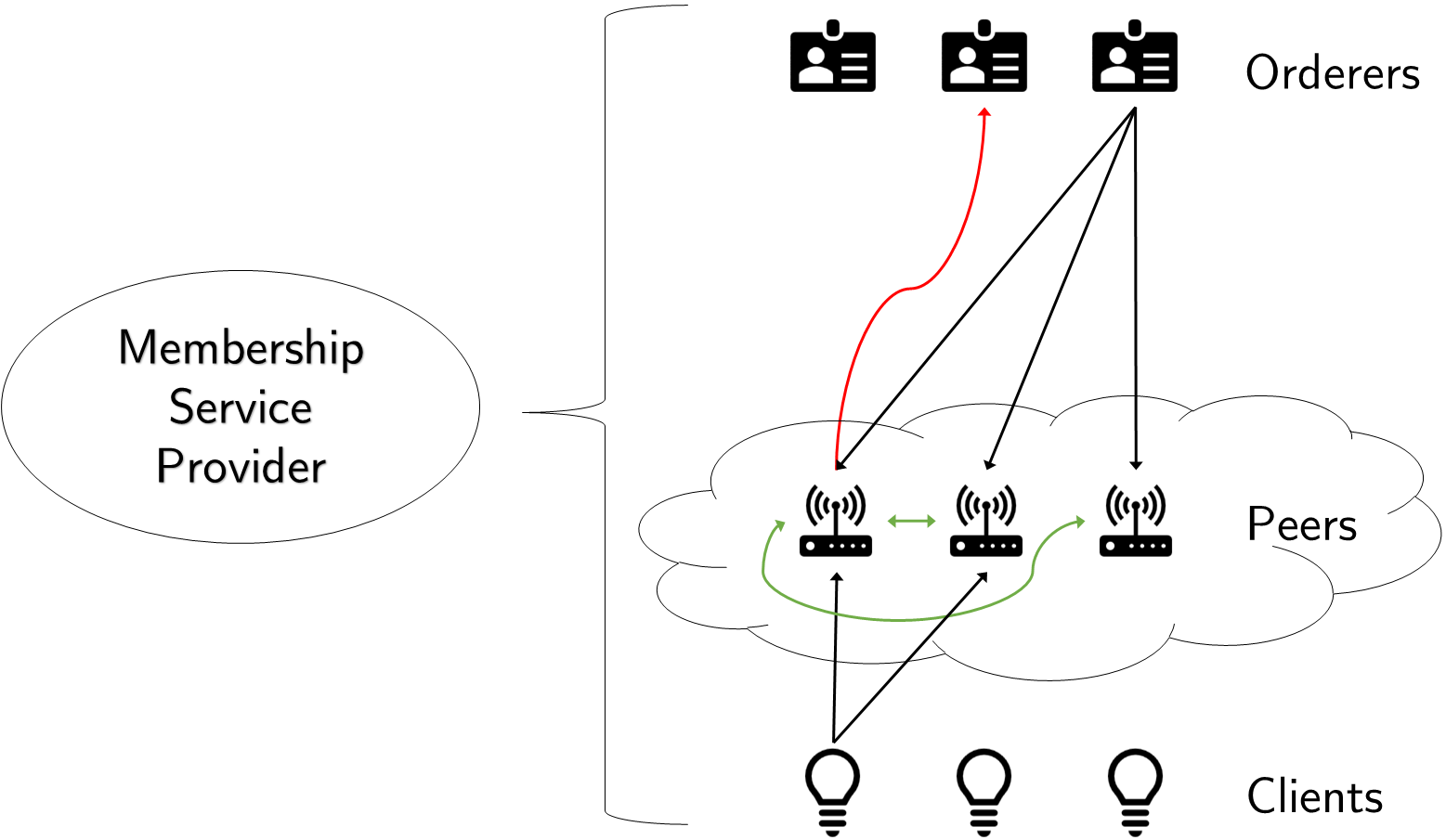} \label{HLF-fig-mod}
	}
	\caption{Modified Hyperledger Fabric architecture.}
	\iffull
	\else
	\vspace{-0.15in}
	\fi
\end{figure*}

\iffull
\subsection{Hyperledger} 

One of the most promising examples of permissioned blockchains is the
Hyperledger project, established by the Linux Foundation and actively
supported by companies such as IBM and
Intel \cite{Hyperledger-architecture-vol1}. The Hyperledger project
aims to satisfy a wide range of business blockchain requirements, and
has developed several frameworks with different implementation
approaches, such as Hyperledger Fabric, Indy, Iroha and Sawtooth. Each
framework uses a different consensus algorithm, or even supports
``pluggable" (rather than hardcoded) consensus like Hyperledger
Fabric~\cite{DBLP:journals/corr/abs-1801-10228}.
To make pluggable consensus possible, Hyperledger Fabric introduces
the ``execute-order-validate" paradigm in its architecture, instead of
the traditional
``order-execute"~\cite{DBLP:journals/corr/abs-1801-10228}. In this
paradigm, the ``maintainer" participants are decoupled from the
``consensus" participants (called \emph{Peers} and \emph{Orderers}
respectively as we will see below), which eventually leads to
satisfying dynamic membership and scalability.

In the following we focus on Hyperledger Fabric which seems to fit
best in our system and provide a high-level description of its
architecture (shown in Figure \ref{HLF-fig}).
Its main components are categorized as follows:
\begin{enumerate}
	\item \textbf{Clients} are responsible for creating a transaction and submitting it to the peers for signing. After collecting a sufficient number of signatures (as defined by the system policy), they submit their transaction to the orderers for including it in a block. Client authentication is delegated to the application.
	\item \textbf{Peers} are the blockchain maintainers, and are also responsible for endorsing clients' transactions. Notice that in the context of Hyperledger, ``Endorsing'' corresponds to the process of applying message authentication.
	\item \textbf{Orderers} collectively form the ``ordering service" for the system. After receiving signed transactions from the clients, the service establishes \emph{consensus} on total order of a collected transaction set. Then the ordering service by delivering blocks to the peers, and ensures the consistency and liveness properties of the system.
	\item The \textbf{Membership Service Provider} ($\msp$) is responsible for granting participation privileges in the system.
\end{enumerate}

In its initial version v0.6, Hyperledger Fabric used the Byzantine fault-tolerant PBFT consensus algorithm~\cite{Hyperledger-fabric-consensus-v0.6}, which only supports static membership for the ordering service participants. It's current version (v2.2)  offers the \textit{Raft}~\cite{DBLP:conf/usenix/OngaroO14} consensus algorithm, which provide crash fault tolerance but not Byzantine fault tolerance, thus preventing the system from reaching consensus in the case of malicious nodes. Hyperledger Fabric could potentially use BFT-SMART in the future~\cite{DBLP:conf/dsn/BessaniSA14,DBLP:conf/dsn/SousaBV18}.

Regarding scalability, although the original version of Hyperledger Fabric has several potential bottlenecks in its architecture, proposals exist to improve its overall performance in terms of operations per second~\cite{DBLP:journals/corr/abs-1901-00910}. These proposals suggest storing blocks in a distributed peer cluster for further scalability improvements. Also while operations (or transactions) per second is one scalability aspect in our setting, the other one is the actual number of peers the system can practically support without heavily impacting its performance. To the best of our knowledge, this has only been experimentally shown for up to 100 peers~\cite{DBLP:journals/corr/abs-1801-10228}. This experiment indicates however a low impact of the number of peers, especially if the network latency is low. Our setting allows such an assumption, since our use-case scenarios are all instantiated in a specific geographical location, as later discussed in section \ref{measurements}. 

\else
\subsection{Modifying Hyperledger Fabric}
\label{mod:hyperledger}

Hyperledger \cite{Hyperledger-architecture-vol1}\iffull \else, a well-known open-source blockchain platform in the permissioned model, \fi satisfies a wide range of business blockchain requirements, and has developed several frameworks, supporting different consensus algorithms or even ``pluggable" (rather than hardcoded) consensus like Hyperledger Fabric~\cite{DBLP:journals/corr/abs-1801-10228}. 
Its main components are categorized as follows:
\begin{enumerate}[leftmargin=*]
	\item \textbf{Clients} are responsible for creating a transaction and submitting it to the peers for signing. After collecting a sufficient number of signatures (as defined by the system policy), they submit their transaction to the orderers for including it in a block. Client authentication is delegated to the application.
	\item \textbf{Peers} are the blockchain maintainers, and are also responsible for endorsing clients' transactions. Notice that in the context of Hyperledger, ``Endorsing'' corresponds to the process of applying message authentication.
	\item \textbf{Orderers} after receiving signed transactions from the clients, establish \emph{consensus} on total order of a collected transaction set, deliver blocks to the peers, and ensure the consistency and liveness properties of the system.
	\item The \textbf{Membership Service Provider} ($\msp$) is responsible for granting participation privileges in the system.
\end{enumerate}

\fi

\iffull
\subsection{Modifying Hyperledger Fabric}
\label{mod:hyperledger}

While Hyperledger Fabric's  \emph{execute-order-validate} architecture offers several advantages discussed previously, we cannot directly use it in our \sysname system, since we assume that lightweight devices (which for Hyperledger Fabric would have the role of ``clients") are limited to only broadcasting data without being capable of receiving and processing. 
\else
Directly using Hyperledger in \sysname is not possible, since we assume that lightweight devices (which for Hyperledger Fabric would have the role of ``clients") are limited to only broadcasting data without being capable of receiving and processing. 
\fi
In Hyperledger Fabric, clients need to collect signed transactions and send them to the ordering service, which is an operation that lightweight devices are typically not capable of performing. 

\noindent \textbf{Our modification.} To address this issue, we propose a modification in Hyperledger architecture. In our modified version, as shown in Figure \ref{HLF-fig-mod}, a client broadcasts its ``transaction" message to all nearby peer nodes. However, the transaction is handled by a \emph{specific} peer (which are equivalent to an aggregator as we discuss in the next section), while peers not ``responsible" for that transaction disregard it. That specific peer then assumes the role of the ``client" in the original Hyperledger architecture simultaneously, while also continuing functioning as a peer node. As a client, it would be responsible for forwarding this transaction to other peers, and collecting the respective signatures, as dictated by the specified system policy, in a similar fashion to original Hyperledger Fabric. 
It would then forward the signed transaction to the ordering service, and wait for it to be included in a block. The ordering service would send the newly constructed block to all peers, which would then append it to the blockchain. 

\noindent \textbf{Security of our modifications.}
The proposed modifications of Hyperledger do not affect the
established security properties (i.e. Consistency and Liveness\iffull
as we define them in Appendix \ref{apdx:consensus} and interpreted
in \cite{DBLP:journals/corr/abs-1801-10228}\fi), since a peer node
simultaneously acting as a client, can only affect the signing process
by including a self-signature in addition to other peers'
signatures. However, because the signing requirements are dynamically
dictated by the system policy, these could be easily changed to
require additional signatures or even disallow self-signatures to
prevent any degradation in security. 
We also note that while this modified version of Hyperledger effectively becomes agnostic to the original client, which otherwise has no guarantees that its broadcasted transaction will be processed honestly, our threat model discussed in the next section captures the above trust model.

\section{\sysname System properties}
\label{properties}
In \sysname there are five main types of participants, most of them inherited by Hyperledger Fabric: the $\msp$, orderers, local administrators, aggregators and sensors. Aggregators are equivalent to \emph{peers} and sensors to \emph{clients} in our modified Hyperledger Fabric architecture discussed in the previous section. We provide a high level description of each participant's role in the system and include detailed definitions in \iffull Appendix \ref{sec-model-defs} \else the full version of our paper~\cite{fullversion}\fi.
\begin{itemize}[leftmargin=*]
	\item The $\msp$ is a trusted entity who grants or revokes authorization for orderers, local administrators and aggregators to participate in the system, based on their credentials. It also initializes the blockchain and the system parameters and manages the system configuration and policy.
	\item Orderers (denoted by $\orderer{}$) receive signed transactions from aggregators. After verifying the transactions as dictated by the system policy they package them into blocks. An orderer who has formed a block invokes the consensus algorithm which runs among the set of orderers $\ordererset$. On successful completion, it is transmitted back to the aggregators with the appropriate signatures.
	\item Local administrators (denoted by $\localadmin{}$, are lower-level system managers with delegated authority from the $\msp$. Each $\localadmin{}$ is responsible for creating and managing a local device group $\iotGroup{}$, which includes one or more aggregators and sensors. He grants authorization for aggregators to participate in the system with the permission of the $\msp$. He is also solely responsible for granting or revoking authorization for sensors in his group, using aggregators to store their credentials.
	\item Aggregators (denoted by $\aggr{}{}$) are the blockchain maintainers. They receive blocks from orderers and each of them keeps a copy of the blockchain. They store the credentials of sensors  belonging in their group and they pick up data broadcasted by sensors. Then they create blockchain ``transactions" based on their data (after possible aggregation), and periodically collect signatures for these transactions from other aggregators in the system, as dictated by the system policy. Finally, they send signed transactions to the ordering service, and listen for new blocks to be added to the blockchain from the orderers. 
	\item Sensors (denoted by $\sens{}{}$) are resource-constrained devices. They periodically broadcast signed data blindly without waiting for any acknowledgment. They interact with local administrators during their initialization, while their broadcasted data can potentially be  received and authenticated by multiple aggregators.
\end{itemize}

We then define the security 
and operational 
properties of  \sysname, in accordance with evaluation principles adopted in  \cite{coindesk2018,DBLP:journals/corr/DorriKJ16,10.1007/978-3-319-28472-9-9,Shafagh:2017:TBA:3140649.3140656}.

\subsection{Threat model \& Assumptions}
\label{threatmdl}
\noindent \textbf{Physical layer attacks and assumptions.}  While our
system cannot prevent physical tampering with sensors that might
affect data correctness, any data discrepancies can be quickly
detected through comparisons with adjacent sensors given the
blockchain immutability
guarantees\iffull~\cite{EPRINT:WusGer17}\fi. Similarly, any malicious
or erroneous data manipulation by an aggregator will result in
detectable discrepancies even when one of the aggregators is not
compromised simultaneously. Of course, if all aggregators become
compromised instantaneously, which is hard in a practical setting, our
system will not detect any discrepancies. This raises the bar
significantly for an adversary who might not be aware or even gain
access to all aggregator nodes at the same time. Finally, attacks such as flooding/jamming and broadcast interception attacks
are out of scope in this paper. 

\noindent \textbf{Trust Assumptions.} We assume that $\msp$ is honest
during system bootstrapping only, and that device group participants
(Local administrators, aggregators and sensors) may behave unreliably
and deviate from protocols. For instance, they might attempt to
statically or dynamically interfere with operations of honest system
participants (e.g. intercept/inject own messages in the respective
protocols), even colluding with each other to do so. This behavior is
expected which our system is designed to detect and thwart.

\noindent \textbf{Consensus Assumptions.}  As in Hyperledger, we
decouple the security properties of our system from the consensus
ones. For reference, this implies tolerance for up to 1/3 Byzantine
orderer nodes, with a consensus algorithm satisfying at least the
fundamental and additionally required properties discussed in Section
\ref{backgrnd-prelims}.

Given the above adversarial setting, we define the following security properties\iffull\footnote{We do not consider data confidentiality in our system, however as discussed later our model could be extended to satisfy confidentiality as well.}\fi:

\newlist{UR}{enumerate}{1}
\setlist[UR]{label=S-\arabic*,ref=S-\arabic*}

\begin{UR}
	\item \label{partic-auth}
	Only authenticated participants can participate in the system. Specifically:
	\begin{enumerate}[label=\alph*.,ref=\ref{partic-auth}\alph*.]
		\item \label{partic-auth-a} An  orderer non-authenticated by the  $\msp$ is not able to construct blocks (i.e., successfully participate in the consensus protocol). The ordering service can tolerate up to $f$ malicious (byzantine) orderers.
		\item \label{partic-auth-b} An  $\localadmin{}$ non-authenticated by the $\msp$ is not able to form a device group $\iotGroup{}$.
		\item \label{partic-auth-c} If an aggregator is not authenticated by the $\msp$, then its signatures on transactions cannot be accepted or signed by other aggregators. 
	\end{enumerate}
	\item \label{sensor-health} \textit{Sensor health:} 
	Sensors are resilient in the following types of attacks:
	\begin{enumerate}[label=\alph*.,ref=\ref{partic-auth}\alph*.]
		\item Cloning attacks: A non-authenticated sensor cannot impersonate an existing sensor and perform operations that will be accepted by aggregators.
		\item Message injection - MITM attack: A malicious adversary cannot inject or modify data broadcasted by sensors.
	\end{enumerate}
	
	\item \label{partic-malic} \textit{Device group safety:} 
	Authenticated participants in one group cannot  tamper with other groups in any way, i.e.:
	\begin{enumerate}[label=\alph*.]
		\item An $\localadmin{}$ cannot manage another  group, i.e. add or revoke participation of an aggregator or sensor in another device group, or interfere with the functionalities of existing aggregators or sensors at any time.
		\item An aggregator (or a coalition of aggregators) cannot add or remove any sensor in device group outside of their scope, or interfere with the functionalities of existing aggregators or sensors at any time.
		\item A sensor (or a coalition of sensors) cannot interfere with the functionalities of existing aggregators or other sensors at any time.
	\end{enumerate}
 
	\item \label{sec-nonrepud} \textit{Non-repudiation and data provenance:} Any \sysname node cannot deny sent data they signed. For all data stored in \sysname , the source must be identifiable.
	\item \label{sec-dos-resil} \textit{DoS resilient:} \sysname  continues to function even if $\msp$ is offline and not available, or an adversary prevents communication up to a number of orderers (as dictated by the consensus algorithm), a number of aggregators (as dictated by the system policy) and up to all but one sensor. Also an adversary is not able to deny service to any system node (except through physical layer attacks discussed before).
	\item \label{sec-pol-conf} \textit{System policy and configuration security:} \sysname policy and configuration can only be changed by $\msp$.
	\item \label{sec-revoc} \textit{Revocation:} The system is able to revoke authentication for any system participant, and a system participant can have its credentials revoked only by  designated system participants.
\end{UR}

\section{Constructions}

\label{constr}
\label{mitm}

We first set the notation we will be using throughout the rest of the paper. By $\secpar$ we denote the security parameter. By $b \leftarrow B(a)$ we denote a probabilistic polynomial-time (PPT) algorithm $B$ with input $a$ and output $b$. By $:=$ we denote  deterministic computation and by $a \rightarrow b$ we denote assignment of value $a$ to value $b$. \iffull We denote a protocol between two parties $A$ and $B$ with inputs $x$ and $y$ respectively as $\{\mathsf{A}(x) \leftrightarrow \mathsf{B}(y)\}$. \fi
By $(\publickey{}{},\secretkey{}{})$ we denote a public-private key pair. 
\iffull A list of elements is denoted by $ [~]$. We denote a block $\block{}$ being appended on a blockchain $\blockchain$ as $\blockchain ||\block{}$.\else We denote concatenation as $||$.\fi

\subsection{Our Hash-based Signature Scheme}
\label{our-primitive}

\iffull Our construction is inspired by Lamport passwords~\cite{Lamport:onetime} and TESLA~\cite{848446,Perrig02thetesla} but \else Our construction is a digital signature scheme that only requires hashing as the main operation. While  inspired by the Lamport passwords~\cite{Lamport:onetime} and TESLA~\cite{848446,Perrig02thetesla}, it \fi \emph{avoids the need for any synchronization} between senders and receivers which is a strong assumption for the IoT setting. 
Instead, we assume the existence of a constant-sized state for both the sender and receiver between signing operations. Our scheme allows for a fixed number of messages to be signed, and has constant communication and logarithmic computation and storage costs 
under the following requirements and assumptions:
\begin{itemize}[leftmargin=*]
	\item There's \emph{no} requirement for time synchronization, and a verifier should only need to know the original signer's $\publickey{}{}$.
	\item The verifier should immediately be able to verify the authenticity of the signature (i.e. without a ``key disclosure delay" that is required in the TESLA family \iffull protocols, described in more detail in Section \ref{prel:onetime-sigs} ). \else protocols. \fi
	\item Network outages, interruptions or ``sleep'' periods can be resolved by requiring computational work from the verifier,  proportional to the length of the outage.
	\item We do not protect against Man-in-the-Middle attacks in the signature level, instead, we use the underlying blockchain to detect and mitigate such attacks as we discuss later in Section \ref{constr:sec-analysis}.
	\item The signer has very limited computation, power and storage capabilities, but can outsource a computationally-intensive pre-computation phase to a powerful system.
\end{itemize}

\renewcommand{\figurename}{Construction}
\setcounter{figure}{0}
\begin{figure}
	
	\fbox{\begin{minipage}{0.96\linewidth}

			\begingroup
			\fontsize{8pt}{12pt}\selectfont
			
			Let $h: \{0,1\}^{*} \rightarrow \{0,1\}^{\secpar}$ be a preimage resistant hash function.
			\begin{itemize}[leftmargin=*]
				\item[] $\otkeygen{\publickey{}{}}{\secretkey{n}{}}{s_{0}}{n}$ 
				\begin{itemize}
					\item sample a random ``private seed" $\privateSeed \leftarrow\{0,1\}^{*}$
					\item generate hash chain $\publickey{}{} = \firstPublic= h(k_{1})= h(h(k_{2})) = ... = h^{i}(k_{i}) = h^{i+1}(k_{i+1}) = ...= h^{n-1}(k_{n-1}) = h^{n}(\privateSeed)$
					\item hash chain creates $n$ pairs of $(\publickey{i}{},\secretkey{i}{})$ where:
					\\ $(\publickey{1}{},\secretkey{1}{}) = (\firstPublic,k_{1}) = (h(k_{1}),k_{1})$, 
					\\ $(\publickey{2}{},\secretkey{2}{}) = (k_{1},k_{2}) = (h(k_{2}),k_{2})$, 
					\\ ... , 
					\\ $(\publickey{i}{},\secretkey{i}{}) = (k_{i-1},k_{i}) = (h(k_{i}),k_{i})$, 
					\\ ..., 
					\\ $(\publickey{n}{},\secretkey{n}{}) = (k_{n-1},\privateSeed) = (h(\privateSeed),\privateSeed)$
					\item initialize a counter $\mathsf{ctr}=0$, store $\mathsf{ctr}$ and  pairs as $[(\publickey{i}{},\secretkey{i}{})]_{1}^{n}$ to initial state $s_{0}$
					\item output $(\publickey{}{}=\publickey{1}{},\secretkey{n}{},s_{0})$.
				\end{itemize} 
				\textbf{Note:} Choosing to store only $(\publickey{}{},\secretkey{n}{})$ instead of the full key lists introduces a storage-computation trade-off, which can be amortized by the ``pebbling" technique we discuss in this section.
				
				\item[] $\otsign{\secretkey{i}{}}{\secretkey{i-1}{}}{s_{i}}{s_{i-1}}{m}{\sigma}$
				
				\begin{itemize}
					\item parse  $s_{i-1}$ and read $\mathsf{ctr} \rightarrow i-1$ 
					\item compute one-time private key $\secretkey{i}{} = k_{i}$ from $n-i$ successive applications of the hash function $h$ on the private seed $\privateSeed$ (or read $k_{i}$ from $[\secretkey{}{}]_{1}^{n}$ if  storing the whole list)
					\item compute $\sigma = h(m||\publickey{i}{})||\secretkey{i}{} = h(m||k_{i-1})||k_{i} = h(m||h(k_{i}))||k_{i}$
					\item increment $\mathsf{ctr} \rightarrow \mathsf{ctr} + 1$, store it to updated state $s_{i}$
				\end{itemize}
				
				\item[] $\otverify{\publickey{}{},n}{m}{\sigma}$
				\begin{itemize}
					\item parse $\sigma = \sigma_{1}||\sigma_{2}$ to recover $\sigma_{2}=k_{i}$
					\item Output $b = (\exists j<n: h^{j}(k_{i}) =\publickey{}{}) \land (h(m||h(k_{i})) = \sigma_{1})$ 
				\end{itemize}
				\textbf{Note:} The verifier might choose to only store the most recent $k_{i}$ which verified correctly, and replace $\publickey{}{}$ with $k_{i}$ above resulting in fewer hash iterations.
			\end{itemize}
			\endgroup
	\end{minipage}}
	\caption{$n$-length Chain-based Signature Scheme}
	\label{prel:onetime-sigs-constr}
	\iffull
	\else
		\vspace{-0.05in}
		\fi
\end{figure}
\renewcommand{\figurename}{Figure}
\setcounter{figure}{1}

\begin{figure}
	\begin{tikzpicture}[>=stealth']
	{[start chain]
		\node[on chain] (A) {$k_{0}$};
		\node[on chain,join=by {<-,"$h$"},right=of A] (B) {$k_{1}$};
		\node[on chain,join=by {<-,"$h$"},right=of B] (C) {$k_{2}$};
		\node[on chain,join=by {<-,"$h$"},right=of C] (D) {$k_{3}$};
		\node[on chain,join=by {<-,"$h$"},right=of D] (E) {$k_{4}$};
		\node[on chain,join=by {<-,"$h$"},right=of E] (F) {$k_{5}$};
	}
\end{tikzpicture}

	\caption{Key generation for $n = 5$ and seed $k_5$. First signature uses as $\publickey{}{} = k_{0}$ and $\secretkey{}{} = k_{1}$.}
	\label{prel:onetime-sigs-ex}
	\iffull
	\else
		\vspace{-0.15in}
\fi
\end{figure}
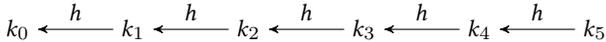

Our scheme, presented in Construction ~\ref{prel:onetime-sigs-constr}, is a chain-based one-time signature scheme\iffull  secure under an adaptive chosen-message attack as formally defined in Definition~\ref{defn:OTS} in the  Appendix\fi, with each key derived from its predecessor as $k_{i} \leftarrow h(k_{i+1})$, $i \in \{n-1, n-2, \ldots, 0\}$ and $h$ is a preimage resistant hash function.
The keys when used in pairs $(k_{i},k_{i-1})$ can be viewed as a public-private key pair for a one-time signature scheme, then forming a one-way hash chain with consecutive applications of $h$. The key $k_{n}$ serves as the ``private seed" for the entire key chain. In the context of integrity, a signer with a ``public key" $k_{i-1} = h(k_{i})$ would have to use the ``private key" $k_{i}$ to sign his message. Since each key can only be used once, the signer would then use $k_{i} = h(k_{i+1})$ as his ``public key" and $k_{i+1}$ as his ``private key", and continue in this fashion until the key chain is exhausted.

For example as shown in Figure \ref{prel:onetime-sigs-ex}, we can construct a hash chain from seed $k_{5}$. For signing the 1st message $m_1$, the signer would use $(\publickey{1}{},\secretkey{1}{}) = (\firstPublic,k_{1})$ and output signature $\sigma = h(m_1||\firstPublic)||k_{1}$. Similarly, for the 2nd message he would use $(\publickey{2}{},\secretkey{2}{}) = (k_{1},k_{2})$ and for the 5th message $(\publickey{5}{},\secretkey{5}{}) = (k_{4},k_{5})$.

Constructing  the one-way hash-chain described above, given the seed $k_{n}$, would require $O(n)$ hash operations to compute $k_{0} = h^{n}(k_n)$, which might be a significant computational cost for resource-constrained devices, as the length of the hash chain $n$ is typically large to offset the constraint of single-use keys. While we could pre-compute all the keys, which would cost a $O(1)$ lookup operation, we would then require $O(n)$ space, which is also a limited resource in such devices.
Using efficient algorithms~ \cite{jakobsson2002fractal,RSA:YSEL09}, we can achieve logarithmic storage and computational costs by placing ``pebbles" at positions $2^j = 1\cdots\left \lceil{log_2(n)} \right \rceil$,  which as shown in Section \ref{measurements-sign-verif} makes our construction practical for resource-constrained devices. The verifier's cost is $O(1)$ when storing the most recently-used $k$.

\iffull
\else
In the full version of the paper~\cite{fullversion} we present formal definitions of chain-based signatures and prove  unforgeability of our scheme.
\fi

\begin{table}[]\centering
	\caption{Hash-based scheme comparison. 
	}
	\label{onewaychain_comparison}
\resizebox{\columnwidth}{!}{%
\begin{tabular}{|l|l|l|l|}
	\hline
	Scheme                                                             & Architecture & NoSync             & NoDelay            \\ \hline
	TESLA~\cite{848446,Perrig02thetesla}                               & Chain        & \xmark & \xmark \\ \hline
	$\mu$TESLA 2-level chain \cite{NDSS:LiuNin03}              & Chain        & \xmark & \xmark \\ \hline
	Sandwich, 1-level, light chain \cite{ACNS:HuJakPer05}                        & Chain        & \xmark & \xmark \\ \hline

	Comb Skipchain \cite{ACNS:HuJakPer05}                        & Chain        & \cmark & \xmark \\ \hline
	
	Short Hash-Based Signatures~\cite{CANS:DahKra09}                   & Chain        & \cmark & \cmark \\ \hline
	XMSS~\cite{PQCRYPTO:BucDahHul11}                                   & Tree         &   \cmark    &   \cmark     \\ \hline
	BPQS \cite{EPRINT:CBHLNS18}                                        &  Chain   &  \cmark       &   \cmark      \\ \hline
	SPHINCS  \cite{EC:BHHLNP15}                                                          & Tree         &  \cmark      & \cmark         \\ \hline
	Our construction                                                   & Chain        & \cmark & \cmark \\ \hline
\end{tabular}
}
\iffull
\else
    \vspace{-0.15in}
    \fi
\end{table}

\begin{table*}[]\centering
	\caption{Hash-based scheme comparison for 256-bit messages and 256-bit security parameter. Sizes in bytes. $\mathbb{M}$,$\mathbb{F}$ and $\mathbb{H}$ denote MAC, PRF and hash operations respectively. $n$ denotes length of chain-based schemes.}
	\label{hashchain_comparison_concrete}
\begin{tabular}{|l|l|l|l|l|l|}
	\hline
	Scheme                                                            & $|\sigma |$       & $|\publickey{}{}|$ & $|\secretkey{}{}|$                                & $\mathsf{Sign()}$                                                      & $\mathsf{Verify()}$                   \\ \hline \cline{6-6} 
	Short Hash-Based Signatures \cite{CANS:DahKra09} & $128 + log_2n$ & $32$          & $64(\left\lceil{log_2(n)} \right \rceil+1)$ &     $(\left\lceil{log_2(n)} \right \rceil +3) \mathbb{H} + 3\mathbb{F} $                                                                   & $\left\lceil{log_2(n)} \right \rceil$ \\ \hline
	
	XMSS \cite{PQCRYPTO:BucDahHul11}                 & 2692 (4963)       & 1504 (68)          & 64                                                & 747$\mathbb{H}$ + 10315$\mathbb{F}$                                    & 83$\mathbb{H}$ + 1072$\mathbb{F}$     \\ \hline
	BPQS \cite{EPRINT:CBHLNS18}                      & 2176              & 68                 & 64                                                & 1073 $\mathbb{H}$                                                      & 1073 $\mathbb{H}$                     \\ \hline

	SPHINCS  \cite{EC:BHHLNP15}    & 41000             & 1056               & 1088                                              & 386$\mathbb{F}$, 385 PRGs, 167519 $\mathbb{H}$                         & 14060 $\mathbb{H}$                    \\ \hline
		Our Construction                                           & 32(64)            & 32                 & 32                                                & $\left\lceil{log_2(n)} \right \rceil \mathbb{H}$                 & 1 $\mathbb{H}$                        \\ \hline 
\end{tabular}
	
\end{table*}

\textbf{Comparison and Discussion.} 
Our scheme is directly comparable with the TESLA Broadcast Message Authentication Protocol~\cite{848446,Perrig02thetesla}, which follows a similar chain-based paradigm but requires some synchronicity between the sender and receiver, and the receiver can only verify a message after some delay. Several other chain-based schemes have been proposed \cite{NDSS:LiuNin03,ACNS:HuJakPer05,CANS:DahKra09}, forming a ``hierarchy'' of chains aiming to improve their efficiency in various aspects. However, most of them do not prevent the synchronicity requirement and delayed verification, in fact some even introduce additional requirements, e.g. special ``commitment distribution'' messages \cite{NDSS:LiuNin03}, where a verifier won't be able to verify a long series of signatures if those are lost.
As our scheme is hash-based, we compare with another family of hash-based signatures schemes that follow a tree structure, e.g. XMSS~\cite{EC:BHHLNP15} and SPHINCS \cite{PQCRYPTO:BucDahHul11}. While these schemes do not have any synchronicity assumptions, their performance is not suited for the low SWaP sensors we consider (even with resource-constrained device optimizations \cite{PKC:HulRijSch16}\iffull which we compare in detail in Appendix \ref{apdx:modsphincs}\fi).
In Table \ref{onewaychain_comparison} we compare with other hash-based schemes in terms of properties (i.e. no synchronicity or delays, denoted as NoSync and NoDelay respectively). In Table \ref{hashchain_comparison_concrete} we provide a concrete comparison with the rest of the schemes satisfying the above properties. In Section \ref{relwork} we discuss some of the above schemes in more detail. 

The caveat in our scheme is that it is susceptible to Man-in-the-Middle attacks. Specifically, an attacker might intercept a signature packet in transit (thus learning the ``ephemeral'' private key) and replace it with an arbitrary message and signature. Nevertheless such attacks are unlikely to be successful in our setting as discussed later in Section \ref{constr:sec-analysis}.

\subsection{Overall \sysname Construction}
\label{bbox-costr}
Our \sysname system consists of the following components as shown in
Figure \ref{secmodelfig-intro} illustrating our modifications to the
Hyperledger Fabric architecture.

\begin{figure}
	\includegraphics[width=0.47\textwidth]{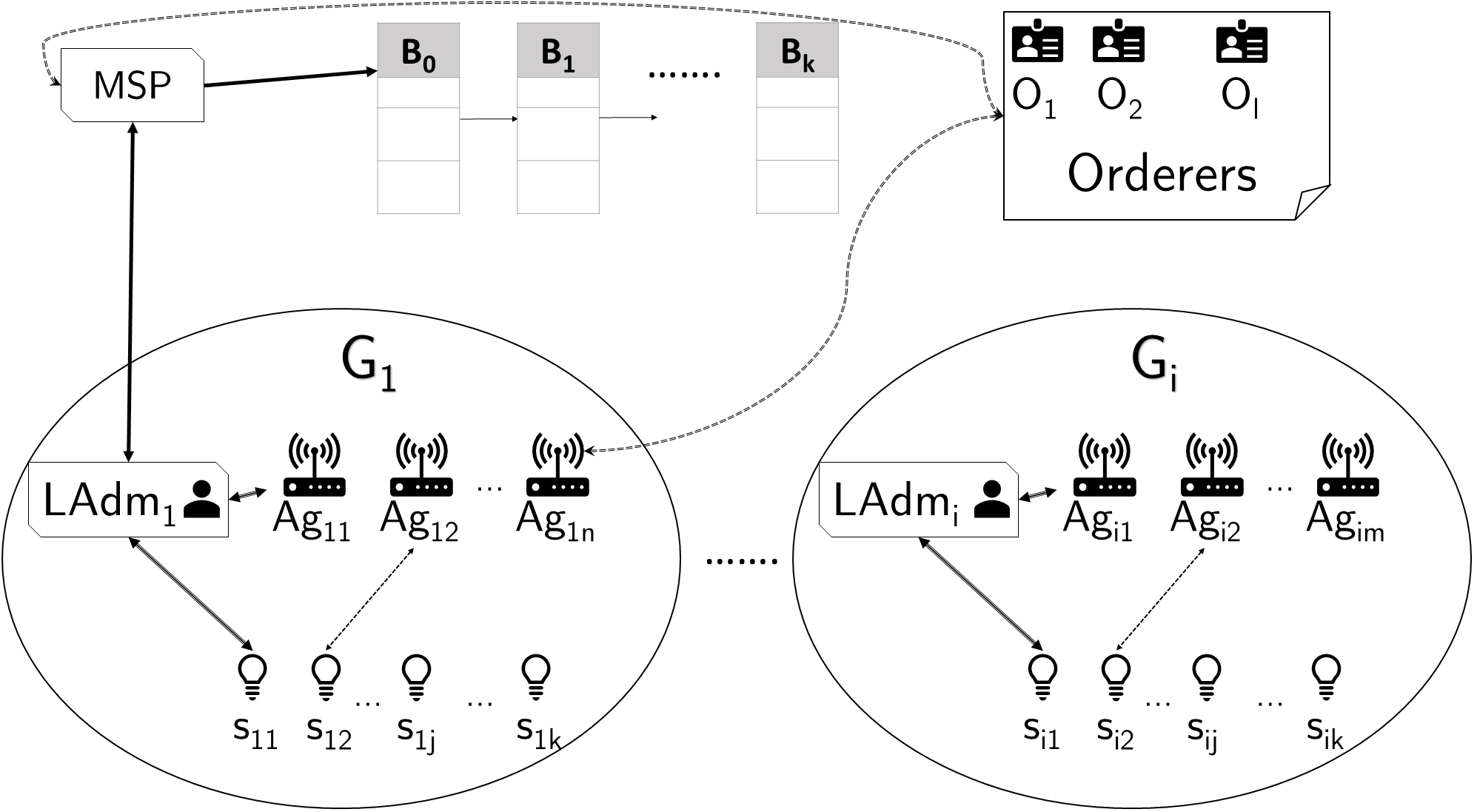}
	\caption{\sysname construction overview} \label{secmodelfig-intro}
	\iffull
	\else
	\vspace{-0.15in}
	\fi
\end{figure}
\begin{itemize}[leftmargin=*]
	\item A (trusted) Membership Service Provider\footnote{The MSP also
  includes the system administrator.} $\msp$, which resembles a
    Trusted Party, and is responsible for authorizing participation in
    the system. The $\msp$ bootstraps the system and forms the genesis
    block, which contains hardcoded information on its public key and
    the consensus algorithm. The genesis block also initializes the
    authorized system participants and the system policy (denoted by
    $\policy$), both of which can be changed later.
	\item A permissioned blockchain $\blockchain$, which consists of
    normal ``transaction" blocks and special ``configuration" blocks.
	\item A configuration $\config$ for $\blockchain$, containing
    membership information for local administrator, orderer and
    aggregators, as well as system policy data. As in Hyperledger
    Fabric, $\config$ is stored in the configuration blocks.
	\item A set of orderer nodes $\ordererset: \{\orderer{1},
    \orderer{2},...,\orderer{\ell} \}$, responsible for achieving
    consensus on forming new blocks. These nodes are assumed static,
    although it can be extended to handle dynamic membership.
	\item A set of device groups $\iotGroupSet: \{\iotGroup{1},
    \iotGroup{2},...,\iotGroup{n} \}$. On each group $\iotGroup{i}$
    there exist:

	  \begin{itemize}
		\item A local administrator $\localadmin{i}$, responsible for its group membership, which includes a set of aggregators and sensors. In order for $\localadmin{i}$ to add or remove an aggregator in the system must also have consent from the $\msp$, however he does not need permission to handle sensor membership.
		\item A set of aggregators $\aggrset{i} : \{\aggr{i}{1}, \aggr{i}{2}, ... , \aggr{i}{m}\}$, which have also the role of \emph{peers} in Hyperledger Fabric. We assume aggregators can perform regular cryptographic operations and aggregate data received from sensors. As discussed in our modified Hyperledger, they also briefly take the role of a ``client". 
		\item A set of sensors $\sensset{i} : \{\sens{i}{1}, \sens{i}{2}, ... ,\sens{i}{k} \} $, which are assumed to be resource-constrained devices. These would be the equivalent of \emph{clients} in the original Hyperledger Fabric architecture, but here they are assumed to only broadcast their data to nearby group aggregators, without expecting a confirmation. The only step where interaction occurs is during initial setup, where they exchange their public key and other initialization data  
		with the group administrator. We also assume that sensors can only perform basic cryptographic operations (i.e. hashing), meaning they can't perform public key cryptography operations that use exponentiations.
		
	\end{itemize}

\end{itemize}

We first describe the initialization process for the system's $\msp$
and genesis block $\block{0}$. After generating its keys, $\msp$
bootstraps the system with pre-populated participation whitelists of
orderers, local group administrators, and aggregators (denoted by
$\ordererlist, \ladmlist$ and $\peerlist$ respectively) and a
pre-defined system policy. Sensors do not need to be tracked from the
$\msp$, as participation authorization for sensors is delegated to the
group local administrators. Local administrators control authorization
privileges with a respective sensor whitelist denoted by
$\sensorList{}$, and they also keep a whitelist of group aggregators
denoted by $\agglistladmin$.

Furthermore, we detail the functionality of reading or updating the system's
configuration, including the permissioned participants and the system
policy. Orderers and local administrators can only be authorized for
participation by the $\msp$, while aggegators need their local
administrator's approval as well. As discussed above, sensor
participation is handled by the local administrators, however, group
aggregators also keep track of group participation for sensors in a
passive manner. The local administrators are also responsible for
revoking participation rights for aggregators and sensors belonging in
their group. In general, granting or revoking participation privileges
is equivalent to adding or removing the participant's public key from
the respective whitelist. \iffull Note that membership verification
can also be handled by accumulators~\cite{EPRINT:BCDLRS17,C:CamLys02}
instead of whitelists to achieve lists of constant size, however we
keep whitelists for simplicity purposes. \fi

Furthermore, on a high-level, sensors ``blindly'' broadcast their data
as signed transactions. Nearby aggregators (belonging to the same
device group) receive and verify the data and collect the required
amount of signatures from other aggregators in the system (as defined
by the system policy), and then submit the signed transaction to the
ordering service.  The orderers then by running the consensus
protocol, ``package'' the collected transactions to form a blockchain
block. Finally, the block is sent back to the aggregators, who as the
blockchain ``maintainers'', append it to the blockchain. The core
system functionalities are shown in Construction \iffull
\ref{constr:algorithms} and we provide a detailed description of all
system algorithms and protocols in Appendix
\ref{sec-model-defs}. \else \ref{constr:algorithms}. \fi.

\renewcommand{\figurename}{Construction}
\setcounter{figure}{1}
\begin{figure}
	
	\fbox{\begin{minipage}{\linewidth}

			\begingroup
			\fontsize{8pt}{12pt}\selectfont
			
			\begin{itemize}[leftmargin=*]
				
				\item[]  $\sensjoin $ 
				\begin{itemize}
					\item Sensor generates a seed uniformly at random, and
            generates hash chain through $\mathsf{OTKeyGen}$
            algorithm. (computation is outsourced to a powerful
            device)
					\item Sensor stores hash chain ``pebbles'' in its memory and
            outputs the last element of the chain as public key to the
            $\localadmin{}$
				\end{itemize}

				\item[]  $\senssenddata $ 
				\begin{itemize}
					\item  Sensor computes signature $\sigma$ for broadcasted data $m$ using $\mathsf{OTSign}$ algorithm 
					
					\item $\sens{i}{j}$ broadcasts $\sigma$ to aggregators in group.
					\item Each aggregator after verifying the signature through $\mathsf{OTVerify}$, checks if any other aggregator received a conflicting message. It adds the message - signature pair in its local state, pending for blockchain submission.
				\end{itemize}
			
				\item[] $\aggsendtx$ 
				\begin{itemize}
					\item Aggregator parses its local state for pending blockchain operations as a transaction.
	
					\item Aggregator computes signature on transaction and sends it to other aggregators.

					\item Each aggregator after verifying signature and sender membership in the system, signs the transaction.
					
					\item The sending aggregator submits signed transaction to ordering service after reaching necessary number of signatures, as dictated by system policy.

					\item Each orderer after verifying signatures, runs consensus algorithm which outputs a blockchain update operation.
					
					\item The blockchain operation is received by orderers who update the blockchain state.
				\end{itemize}
			
				\item[] $\mathsf{SensorTransfer}$ 
				\begin{itemize}
					\item Aggregator encrypts the state for the sensor under the reveiving aggregator's $\publickey{}{}$ (i.e. the most recent received $\secretkey{i}{}$) and submits it to the blockchain using $\aggsendtx$. Sensor is removed from the device group and is transferred to new group.
					\item Receiving aggregator decrypts state from the blockchain and resumes verification of received data from sensor.
				\end{itemize}
			\end{itemize}
			\endgroup
	\end{minipage}}
	\caption{\sysname core algorithms and protocols}
	\iffull
	\else
\vspace{-0.15in}
\fi
	\label{constr:algorithms}
\end{figure}

\renewcommand{\figurename}{Figure}
\setcounter{figure}{3}

\emph{Sensor join:} Defined by $\sensjoin()$ protocol between a sensor
and a Local administrator. This is the only phase when a sensor is
interacting with the system, as the $\localadmin{}$ generates a new
hash chain and its associated pebbles in a powerful device. The
pebbles are then loaded to the sensor, and $\localadmin{}$ updates the
group aggregators with the new sensor's public key.

\emph{Sensor broadcast:} Defined by $\senssenddata()$ protocol between
a sensor and group aggregators. For some data $m$, the sensor computes
the one-time hash-based signature using $\mathsf{OTSign}()$ and the
signed data $m,\sigma$ is broadcasted to all group aggregators. If
there are any aggregator who receives a different signed message
$m',\sigma$, the message is discarded, else it remains in the
aggregator's pending memory for processing.

\emph{Aggregator transaction:} Defined by $\aggsendtx()$ protocol
between aggregators and orderers. For an aggregator to submit
aggregated data to the blockchain, it first needs to collect the
needed signatures from other aggregators. Then it submits the signed
transaction to the ordering service, which in turn executes the
$\consensus()$ algorithm to construct a block with a set of signed
transactions. Finally, the block is transmitted to the aggregators, who
append the block as the blockchain maintainers.

\emph{Sensor transfer:} Defined by
$\mathsf{SensorTransfer}$ algorithm, executed when a sensor is
transferred to a new location or device group.  The handing-over
aggregator saves its state of our signature scheme w.r.t. that sensor
and encrypts it on the blockchain under the receiving aggregator's
public key. After sensor transfer, the receiving aggregator decrypts
that state and resumes message verification.

Optionally in our construction, a symmetric group key $\groupKey{}$
can be shared between each group's local administrator, aggregators
and sensors for confidentiality purposes. However, the additional
encryption operations have an impact mainly on sensors, which have
constrained computational and storage resources. Note that using such
key for authentication or integrity would be redundant since these
properties are satisfied using public keys existing in the appropriate
membership lists and revocation operations can still be performed at
an equivalent cost using those lists.

\subsection{Security Analysis}
\label{constr:sec-analysis}
\iffull
\begin{thm} [informal]
	The construction in Section \ref{bbox-costr} satisfies 
	participant authentication (\ref{partic-auth}), sensor health \ref{sensor-health} and device group safety properties (\ref{partic-malic}) assuming  $(\mathsf{SignGen}$, $\mathsf{Sign}$, $\mathsf{SVrfy})$ is an existentially unforgeable under a chosen-message attack signature scheme, $(\mathsf{OTKeyGen},\mathsf{OTSign},\mathsf{OTVerify})$ is an unforgeable one-time chain based signature scheme, $\msp$  is honest and not compromised and the consensus scheme $(\mathsf{TPSetup}$, $\mathsf{PartyGen}$, $\mathsf{TPMembers}, \consensus)$  satisfies the consistency property. 
\end{thm}

\noindent \textit{Proof Sketch.}
We now provide a proof sketch arguing about the security of our scheme.

\noindent \textit{\ref{partic-auth} Participant Authentication.} 
We require that only authenticated participants can participate in the different functions of our protocol. We argue that if an 
adversary breaks the participant authentication property then it could break unforgeability of the underlying signature scheme.
Specifically:
\begin{enumerate}
	\item For property \ref{partic-auth-a}, in protocol $\ordereradd$ (coupled with $\updateconfig$) the use of an unforgeable signature scheme guarantees that no one but the MSP can authenticate orderers, while in protocol $\consensus$ the same scheme guarantees that only the authenticated orderers can perform this core functionality. Also recall from the previous property that an adversary being able to authenticate orderers could break the immutability property. 
	\item For property \ref{partic-auth-b}, in protocol $\localadmreg$ (coupled with $\updateconfig$) the use of an unforgeable signature scheme guarantees that no one but the MSP can authenticate Local Administrators, while $\aggsetup, \aggAdd, \aggupd$, $\senssenddata$ and $\grprem$  guarantee that only the authenticated local administrators can perform these functionalities.
	\item For property \ref{partic-auth-c}, in protocol $\aggAdd$ (coupled with \\$\updateconfig$) the use of an unforgeable signature scheme guarantees that no one but the MSP can authenticate Aggregators, while $\aggupd$ and $\aggsendtx$  the same scheme guarantees that only the authenticated aggregators can perform these functionalities.
\end{enumerate}

\noindent \textit{\ref{sensor-health} Sensor health.} In order for an
adversary to impersonate/clone a sensor, it would either have
to break the unforgeability of our signature scheme, or launch a MITM
attack which is a potential attack vector as discussed in Section
\ref{our-primitive}. 

As discussed in Section \ref{threatmdl}, we consider jamming attacks
at the physical layer outside the scope of this paper. Given the
nature of our setting where a sensor's broadcast has typically short
range, we consider MITM and message injection attacks hard and
unlikely to launch but we still consider them as part of our threat
model. Even in these unlikely scenarios, MITM attacks can be easily
mitigated in \sysname. A first approach for detecting such attacks is
to leverage blockchain properties, where aggregators can compare data
received from a sensor in the blockchain level. Our assumption here is
that sensor data can be received by more than one aggregators in the
vicinity of the sensor which is a reasonable senario for typical dense
IoT deployments. If there's even one dissenting aggregator, probably
victim of a MITM attack, all the associated data would be considered
compromised and disregarded and the operator will be notified of the
data discrepancy detected. The above approach while simple, still
permits a MITM attacker to ``eclipse'' a sensor from the system using
a jamming attack.

An alternative approach is to make a proactive check in a group level,
where each aggregator would verify the validity of its received data
by comparing it with other aggregators before even submitting it to
the blockchain. In both above strategies, the attacker's work
increases significantly because he would need to launch simultaneous
MITM attack between the sensor and all aggregators in the vicinity. We
adopt the second approach in our Construction in Appendix
\ref{sec-model-defs}.

The above properties and strategies ensure that only data broadcasted
by authenticated sensors are accepted by aggregators in
$\senssenddata$.

\noindent \textit{\ref{partic-malic} Device group safety.} An
adversary wanting to break device group safety would either
have to add or revoke aggregators or sensors in an existing group
through $\aggAdd$, $\aggupd$ or $\grprem$ thus breaking unforgeability
of the signature scheme used in these protocols or interfere with
existing authenticated sensors in a group through $\senssenddata$ by
breaking unforgeability of the one-time chain based signature scheme.
\qed

Integrity, non-repudiation and data provenance requirements
(\ref{sec-nonrepud}) are core properties of
any digital signature scheme thus directly satisfied in \sysname.

Additionaly we argue that our system is DOS resilient (\ref{sec-dos-resil}) in the following scenarios:
\begin{itemize}
	\item $\msp$ offline or not available: The core system functionality is not affected, although there can be no configuration changes in the system. All algorithms and protocols (except those involving adding or revoking orderers, local administrators or aggregators or those involving system policy changes) perform authentication  through the configuration blocks and not the $\msp$ itself.
	\item Orderers unavailable: Reduces to tolerance properties of consensus algorithm.
	\item  $\localadmin{}$ unavailable: The core system functionality is not affected, although there can be no administrative operations in the respective group. 
	\item $\aggr{}{}$ unavailable: Transactions are not processed only in the respective groups. However if more than $\tau$ aggregators are unavailable as required in $\aggsendtx$, no transactions can be processed in the whole system.
\end{itemize}
Also an adversary might attempt to flood an aggregator by broadcasting messages and arbitrary signatures. In this scenario the aggregator would be overwhelmed since by running $\mathsf{OTVerify}$ for each message-signature pair separately, it would have to check the signature against all hash chain values up to the first public key. To mitigate this we propose checking only for a few hashes back to the chain (defined by a system parameter ``maxVerifications" as shown in Algorithm \ref{measurements-aggregator-pseudocode} in the Appendix). This parameter can be set by the local administrator but should be carefully selected. A small value might result in need of frequent re-initializations for the sensors - if a long network outage occurs between a sensor and an aggregator and they lose ``synchronization", the local administrator should reinitialize the sensor in the device group. On the other hand, a large value would amplify the impact of DOS attacks.  

Policy and configuration security (\ref{sec-pol-conf}) is ensured by algorithms $\updateconfig$ and $\readconfig$, as the first  algorithm creates a special configuration transaction signed by the $\msp$ and the second returns configuration data originating from such a transaction.

Revocation (\ref{sec-revoc}) is made possible by $\noderem$ and $\grprem$ (in conjuction with whitelists used throughout all system protocols and algorithms). Also the unforgeability of the underlying signature scheme ensures that only the $\msp$ (and the $\localadmin{}$ respectively only for aggregators and sensors) can revoke these credentials. 

\noindent \textbf{Remark.} One might suggest to use MACs instead of our proposed signature scheme for sensor authentication. We discuss this in Appendix \ref{constr:notmacs}.

\else

Given the threat model discussed in Section \ref{threatmdl}, most of
the security properties (all but \ref{sensor-health} and
\ref{sec-dos-resil}) rely on the security of the underlying signature
scheme and consensus properties. As it is straightforward to prove
security for these, we focus on \textit{\ref{sensor-health} sensor
  health} security property (which includes resilience to MITM
attacks) and \textit{\ref{sec-dos-resil}} (resilience to DoS attacks).

In order for an adversary $\adv$ to impersonate/clone a sensor, it
would either have to break the unforgeability of our signature scheme,
or launch a MITM attack which is a potential attack vector as
discussed in Section \ref{our-primitive}.

As discussed in Section \ref{threatmdl}, we consider jamming attacks
at the physical layer outside the scope of this paper. Given the
nature of our setting where a sensor's broadcast has typically short
range, we consider MITM and message injection attacks hard and
unlikely to launch but we still consider them as part of our threat
model. Even in these unlikely scenarios, MITM attacks can be easily
mitigated in \sysname. A first approach for detecting such attacks is
to leverage blockchain properties, where aggregators can compare data
received from a sensor in the blockchain level. Our assumption here is
that sensor data can be received by more than one aggregators in the
vicinity of the sensor which is a reasonable senario for typical dense
IoT deployments. If there's even one dissenting aggregator, probably
victim of a MITM attack, all the associated data would be considered
compromised and disregarded and the operator will be notified of the
data discrepancy detected. The above approach while simple, still
permits a MITM attacker to ``eclipse'' a sensor from the system using
a jamming attack.

An alternative approach is to make a proactive check in a group level,
where each aggregator would verify the validity of its received data
by comparing it with other aggregators before even submitting it to
the blockchain. In both above strategies, the attacker's work
increases significantly because he would need to launch simultaneous
MITM attack between the sensor and all aggregators in the vicinity.

Additionaly, we argue that our system is DOS resilient (\ref{sec-dos-resil}) in the following scenarios:
\begin{itemize}[leftmargin=*]
	\item $\msp$ offline or not available: The core system functionality is not affected, although there can be no configuration changes in the system. All algorithms and protocols (except those involving adding or revoking orderers, local administrators or aggregators or those involving system policy changes) perform authentication  through the configuration blocks and not the $\msp$ itself.
	\item Orderers unavailable: Reduces to tolerance properties of the  consensus algorithm.
	\item  $\localadmin{}$ unavailable: The core system functionality is not affected, although there can be no administrative operations in the respective group. 
	\item $\aggr{}{}$ unavailable: Transactions are not processed only in the respective groups. \iffull However, if more than $\tau$ aggregators are unavailable as required in $\aggsendtx$,\else However if the number of unavailable aggregators exceeds a certain threshold, \fi no transactions can be processed in the whole system.
\end{itemize}
Also, an adversary might attempt to flood an aggregator by broadcasting messages and arbitrary signatures. In this scenario, the aggregator would be overwhelmed since by running $\mathsf{OTVerify}$ for each message-signature pair separately, it would have to check the signature against all hash chain values up to the first public key. To mitigate this, we propose checking only for a few hashes back to the chain specified by a parameter (defined by a system parameter ``maxVerifications" as shown in Algorithm \ref{measurements-aggregator-pseudocode}). This parameter can be set by the local administrator but should be carefully selected. A small value might generate the need of frequent re-initializations for the sensors - if a long network outage occurs between a sensor and an aggregator and they lose ``synchronization", the local administrator should reinitialize the sensor in the device group. On the other hand, a large value would amplify the impact of DoS attacks.  

\begin{algorithm}
	\caption{Sensor send data}
	\label{measurements-sensor-pseudocode}
	\begin{algorithmic}[1]
		\STATE tempkey $\leftarrow \firstPublic$
		\STATE initPebbles()
		\WHILE {True}
		\STATE m $\leftarrow$ readSensor()
		\STATE output.type $\leftarrow$ ``payload"
		\STATE output.data $\leftarrow$ m
		\STATE transmit(output)
		\STATE T1.start()
		\STATE hashedData $\leftarrow h($m||tempkey$)$
		\STATE output.type $\leftarrow$ ``hash"
		\STATE output.data $\leftarrow$ hashedData
		\STATE transmit(output)
		\STATE tempkey $\leftarrow$ computePebbles()
		\COMMENT{as in~\cite{jakobsson2002fractal}}
		\STATE output.type $\leftarrow$ ``secretKey"
		\STATE T1.end()
		\STATE output.data $\leftarrow$ tempkey
		\STATE transmit(output)
		\ENDWHILE
		
	\end{algorithmic}
\end{algorithm}

\begin{algorithm}
	\caption{Aggregator receive data}
	\label{measurements-aggregator-pseudocode}
	\begin{algorithmic}[1]
		\STATE publickey $\leftarrow \firstPublic$
		\STATE verifications $\leftarrow$ 0
		\WHILE {verifications $<$ maxVerifications}
		\STATE check1 $\leftarrow$ False
		\STATE check2 $\leftarrow$ False		
		\STATE  read $\leftarrow$ input()
		\IF {read.type = ``payload"}
		\STATE T3.start()	
		\STATE m $\leftarrow$ read.data
		\ELSIF {read.type = ``hash"}
		\STATE $s_{1} \leftarrow$ read.data
		\ELSIF {read.type = ``secretKey"}
		\STATE $s_{2} \leftarrow$ read.data
		\STATE tempkey $\leftarrow s_{2}$
		\STATE $i \leftarrow 0$
		\STATE T2.start()
		\WHILE{ $i <$ maxVerification $\land$ doWhile = True  } 
		\IF { $h($tempkey$)$ = publickey}
		\STATE check1 $\leftarrow$ True
		
		\IF { $h($m||publickey$) = s_{1}$}
		\STATE check2 $\leftarrow$ True
		\STATE{publickey $\leftarrow$ secretkey }
		\ENDIF
		\STATE doWhile = False		 
		\ELSE 
		\STATE{tempkey $\leftarrow$ $h($tempkey$)$ }
		\STATE i++
		\ENDIF
		\ENDWHILE
		\IF {check1 $\land$ check2 = True}
		\STATE print(``Payload m is valid")
		\STATE verifications++
		\STATE T2.end()
		\ELSE 
		\STATE{print("Verification failed")}
		\ENDIF
		
		\STATE T3.end()
		\ENDIF
		
		\ENDWHILE
	\end{algorithmic}
\iffull
\else
	\vspace{-0.05in}
\fi

\end{algorithm}

\fi

\section{Performance Evaluation \& Measurements}
\label{measurements}
\subsection{The IIoT Setting With Constrained Devices}
\label{measurements-scenario}

IIoT environments are complex systems comprising of heterogeneous devices that can be tracked at different organizational layers, namely (a) computational, (b) network, (c) sensor/edge layers \cite{wu2020convergence}. Devices at the higher levels are powerful servers dedicated to the analysis of data, storage, and decision making. They frequently reside outside the factory premises, i.e., in cloud infrastructures.
\begin{table}
	\caption{Classes of Constrained Devices in terms of memory capabilities according to RFC 7228.}
	\label{tab:constrcapab}
	\resizebox{0.5\columnwidth}{!}{
		\begin{tabular}{lll}
			\hline
			Name    & RAM      & Flash     \\ \hline
			Class 0 & $<<$10 KiB & $<<$100 KiB \\
			Class 1 & ~10 KiB  & ~100KiB   \\
			Class 2 & ~50KiB   & ~250KiB  \\ \hline
		\end{tabular}
	}
\iffull
\else
    \vspace{-0.15in}
    \fi
\end{table}
On the other hand, on-site and at the edge layer, a myriad of low-SWaP devices such as sensors and actuators reside, assigned with the tasks of posting their data or reconfiguring their status based on received instructions. On typical real-life IIoT deployments, the processing speed of such devices ranges from tens (e.g., Atmel AVR family) to hundreds of Mhz (e.g., higher-end models of ARM Cortex M series). Diving even deeper, at the lower end of the spectrum, one may observe sensor-like devices that are severely constrained in memory and processing capabilities. 

Such extremely constrained devices have been considered by RFC 7228~\cite{RFC7228} 
which underlines that ``most likely they will not have the resources required to communicate directly with the Internet in a secure manner''. Thus, the communication of such nodes must be facilitated by stronger devices acting as gateways that reside at the network layer. In Table \ref{tab:constrcapab} we provide a taxonomy of constrained devices residing at the edge of IIoT  according to RFC 7228. 

In this work, we consider a generic IIoT application scenario that involves Class 0 devices which are connected to more powerful IoT gateways in a sensor/gateway ratio of 10:1.
The chosen platforms and all experimental decisions were made to provide a realistic scenario under the following assumptions: (a) devices severely constrained in terms of computational power and memory resources (Class 0) and (b) moderately 
demanding in terms of communication frequency (i.e. transmission once every 10 seconds).

\subsection{Evaluation Setup}

Our testbed consists of Arduino UNO R3~\cite{arduino-uno-rev3} open-source microcontroller boards equipped with ATmega328P 16 MHz microcontroller and 2KB SRAM fitted with a Bluetooth HC-05 module. These devices are really constrained and they represent the minimum of capabilities in all of IoT sensors utilized in our experimental scenarios (Class 0 in Table \ref{tab:constrcapab}).
For the gateways, we use Raspberry Pi 3 Model B devices equipped with a Quad Core 1.2GHz  BCM2837 64bit CPU and 1GB RAM.

We first focus on evaluating our system in a device group level\footnote{Our code is available at \url{https://github.com/PanosChtz/Black-Box-IoT}}. 
We use the one-time signature scheme outlined in Construction~\ref{prel:onetime-sigs-constr} and SHA256 as the hash function $h()$. 
The length of the hash chain \iffull as defined in section \ref{prel:onetime-sigs} \fi sets the upper bound on the number of one-time signatures each sensor $\sens{i}{}$ can generate. In the case where the sensor's available signatures are depleted, it would enter an ``offline" state and the Local Administrator $\localadmin{}$ would need to manually renew its membership in the system through the $\sensjoin$ protocol. In a large-scale deployment of our system however, frequent manual interventions are not desirable, so our goal is to pick a sufficiently large $n$ such that the available one-time signatures to the sensor last for the sensor's lifetime.
As discussed above and taking similar schemes' evaluations into account~\cite{TCHES:AmiCurZbi18}, we consider a frequency of one (1) signing operation per 10 seconds for simplicity. We consider sensor lifetimes between 4 months as an lower and 21 years as a upper estimate (as shown in Table \ref{measurements-table}), which imply a hash chain between $2^{20}$ and $2^{26}$ elements respectively.

In the setup phase, we pre-compute the hash-chain as needed by the pebbling algorithm~\cite{jakobsson2002fractal} and load the initial pebble values into the sensor. We first measure the actual needed storage on the sensor for various values of $n$. Note that for $n=2^{26}$, the lower bound for needed storage using a 256-bit hash function is about $26 \cdot 256 =$ 832 bytes of memory. Then we set the sensor device to communicate with the aggregator through Bluetooth in broadcast-only mode and measure the maximum number of signing operations that can be transmitted to the aggregator for various values of $n$, as well as the verification time needed on the aggregator side since it will need to verify a large number of sensor messages. The fact that we are able to run \sysname on Class 0 devices demonstrates the feasibility of our approach for all low-SWaP sensors.

\subsection{Signing and Verification}
\label{measurements-sign-verif}

We run our experiments under different scenarios and multiple times. Our evaluation results, which are shown in Table \ref{measurements-table}, represent the statistical average across all measurements. Note that for measuring the average signature verification time on the aggregator side, we assume that the aggregator is able to receive all the data broadcasted by the sensor. If a network outage occurs between them (and the sensor during the outage keeps transmitting), the aggregator after reestablishing connection would have to verify the signature by traversing the hash chain back up to the last received secret key, which incurs additional computation time (in Figure \ref{hash-vs-ecdsa-outage} we show the associated verification cost in such occasions). As expected, the verification time is relatively constant in all measurements, about 0.031ms on average. This suggests that such an aggregator could still easily handle $10^{5}$ sensors transmitting data for verification (as we considered one transmission every 10 seconds for each sensor).

Table \ref{measurements-table}, shows that the pebbles data stucture consumes most of the required memory storage in our implementation, while the remaining program requires a constant amount of memory for any number of pebbles. 
We also observe a slight impact of the number of pebbles on the total verification time, which is mainly affected by the sensor's capability to compute the signature on its message and the next secret key. For example, the sensor needs 50ms to compute the next signature with $n=2^{26}$ and 49.95ms for $n=2^{24}$. Also by comparing the total verification time with the signature computation time, we conclude the extra 14.3 msec are needed for transmitting the signature.

\begin{table}[]
		\caption{Evaluation for sensor-aggregator protocol - Average verification times}
	\label{measurements-verif-timers}
	\resizebox{\columnwidth}{!}{%
	\begin{tabular}{l|l|l|l|l|l|l|l|l|}
		\cline{2-9}
		& \multicolumn{4}{c|}{T2 ($\mu$sec)} & \multicolumn{4}{c|}{T3 (msec)} \\ \hline
		\multicolumn{1}{|l|}{maxV} & 20       & 22      & 24      & 26      & 20      & 22     & 24     & 26     \\ \hline
		\multicolumn{1}{|l|}{100}              & 28.12    & 31.18   & 31.34   & 28.95   & 42.83   & 42.84  & 42.91  & 43.08  \\ \hline
		\multicolumn{1}{|l|}{500}              & 30.78    & 31.94   & 30.31   & 30.63   & 51.25   & 51.23  & 51.37  & 51.39  \\ \hline
		\multicolumn{1}{|l|}{1000}             & 31.39    & 30.96   & 31.14   & 30.74   & 55.27   & 55.35  & 55.36  & 55.41  \\ \hline
		\multicolumn{1}{|l|}{2500}             & 30.57    & 30.97   & 32.39   & 30.86   & 60.61   & 60.65  & 60.7   & 60.78  \\ \hline
		\multicolumn{1}{|l|}{5000}             & 33.26    & 31.7    & 31.66   & 31.43   & 64.66   & 64.74  & 64.79  & 64.83  \\ \hline
		\multicolumn{1}{|l|}{10000}            & 33.34    & 33.38   & 33.6    & 31.41   & 68.68   & 68.75  & 68.78  & 68.86  \\ \hline
	\end{tabular}
}
\iffull
\else
	\vspace{-0.15in}
	\fi
\end{table}

In Table \ref{measurements-verif-timers} we provide a series of measurement results for the average verification time of 1 signature on the aggregator. By T2 we denote the verification time of a signature and by T3 the total verification time by an aggregator \iffull (we provide detailed algorithms for our measurements in Appendix \ref{apdx:evaldetails}.) \else(as shown in Algorithm \ref{measurements-aggregator-pseudocode}).
\fi
The average  total verification time (denoted by maxV) increases significantly as we require more verification operations from the Arduino device. This happens because of dynamic memory fragmentation as the pebbling algorithm updates the pebble values.

\paragraph{Comparison with ECDSA} We compare our lightweight scheme with ECDSA, which is commonly used in many blockchain applications. We assume IoT data payloads between 50 and 220 bytes,  which can accommodate common data such as timestamps, attributes, source IDs and values. In Table \ref{sign-verify-costs} we show that our scheme is more efficient compared to ECDSA by 2 and 3 orders of magnitude for signing and verification respectively. Even when considering larger payload sizes which impact hash-based signature operations, 
our scheme remains much more efficient. However, verification cost for our scheme  increases linearly during network outages, and as shown in Figure \ref{hash-vs-ecdsa-outage} it might become more expensive than ECDSA when more than 2400 signature packets are lost. 

Another metric we consider is energy efficiency, which is of particular importance in IoT applications that involve a battery as power source. Our experiments depicted in Figure \ref{energy-efficiency} show that our ATmega328P microcontroller can perform more than 50x hash-based signing operations compared to the equivalent ECDSA operations for the same amount of power. 
Finally, while our hash-based signature normally has a size of 64 bytes (as shown in Table \ref{measurements-table}), we can ``compress'' consecutive signatures along a hash chain to 32 bytes by only publishing the most recent $k_{i}$. The verifier would then generate the previous hash chain values at a minimal computational cost. This makes possible to store more authenticated data in the blockchain, as we show below.

\begin{table}[t ]
	\caption{Evaluation for sensor-aggregator protocol (average values for 5000 verifications)}
	\label{measurements-table}
	\begin{tabular}{|p{4cm}|c|c|c|c|}
		\hline
		& & & &     \\[-2.5ex]
		Hash Chain length $n$ \xdef\tempwidth{\the\linewidth} & $2^{20}$   & $2^{22}$   & $2^{24}$    & $2^{26}$     \\ \hline 
		\multicolumn{1}{|m{\tempwidth}|}{Sensor lifetime for 1sig/10sec (m: months, y: years)}           & 4 m   & 16 m   & 5 y    & 21  y   \\  
		\thickhline \thickhline
		Pebble Gen time (seconds)                & 1.62 & 6.49 & 24.57 & 95.33  \\ \hline
		\multicolumn{1}{|m{\tempwidth}|}{Verification time per signature (msec)}             & \multicolumn{4}{c|}{0.031}         \\ \thickhline \thickhline
		Signature size (bytes)     & \multicolumn{4}{c|}{64+ $|m|$}     \\ \hline
		\multicolumn{1}{|m{\tempwidth}|}{Total dynamic memory usage (bytes)}     &  1436    & 1520     & 1604      & 1678   \\ \hline
		\multicolumn{1}{|m{\tempwidth}|}{Pebble struct memory usage (bytes)}     & 840     &   924 & 1008     & 1082   \\ \hline
		Program memory usage (bytes)     & \multicolumn{4}{c|}{596}     \\ \hline
		\multicolumn{1}{|m{\tempwidth}|}{Signature computation time (msec)}   & 49.82     & 49.88   & 49.95     & 50.00      \\ \hline
		\multicolumn{1}{|m{\tempwidth}|}{Average total verification time per signature (msec)} & 64.15   & 64.25  & 64.26   & 64.32   \\ \hline
		Communication cost (msec)    & \multicolumn{4}{c|}{14.3}     \\ \hline
		 
	\end{tabular}
\iffull
\else
\vspace{-0.15in}
\fi
\end{table}

\subsection{Consensus Performance}
\label{measurements-consensus}

Considering the use-case scenario discussed in Section \ref{measurements-scenario}, we discuss the performance of our \sysname system as a whole. We show that the most important metric in the system is the transaction throughput which heavily depends on the ability of the SWaP sensors to transmit data in a group setting. Of course, the scalability of the system overall is also directly proportional to the number of system active participants it can support simultaneously.

\smallskip
\noindent \textit{Sensors.} Our measurements indicate that the aggregator - which is a relatively powerful device - is not the bottleneck in the protocol execution.  Based on the measurements in Table \ref{measurements-table}, we can safely assume that a single aggregator can verify over a thousand sensors' data being continuously broadcasted, since the signature computation time by a sensor is three (3) orders of magnitude larger than the verification time by an aggregator. This is still a pessimistic estimation, since we previously assumed that a sensor broadcasts (and signs) data every 10 seconds, which implies that the aggregator can accommodate even more sensors.

\smallskip
\noindent \textit{Orderers.} Since orderers only participate in the consensus protocol to sign blocks, we only need a few orderers such that our system remains resilient to attacks at the consensus level 
should a subset of orderers become compromised. Orderers can be strategically distributed over a geographical area to minimize the network latency between an aggregator and the ordering service, controlled by the main organization (which also controls the $\msp$). Evaluations performed in previous works have shown that by having 3 orderers, 3000 transactions/second can be easily achieved using the consensus protocol used in the current version of Hyperledger Fabric (with a potential of further improvement in a future adoption of BFT-SMART), and even considering up to 10 orderers in the system does not greatly affect its performance~\cite{DBLP:journals/corr/abs-1801-10228,DBLP:conf/dsn/SousaBV18}.

\smallskip
\noindent \textit{Aggregators.} 
The expected number of aggregators in the system depends on the use case as it is expected. As discussed in Section \ref{measurements-scenario}, where gateways play the role of \sysname aggregators, we consider a sensor/gateway ratio of 10:1 for our evaluation purposes.
To our knowledge, no evaluation of Hyperledger Fabric has ever been performed to consider such a great number of peers, which would require a great amount of resources to perform. However, by adopting the evaluation performed in ~\cite{DBLP:journals/corr/abs-1801-10228} which measured the throughput in terms of number of peers up to 100 (which as discussed, are the aggregators in our system), we can extrapolate this evaluation to the order of thousands, which shows that with the aid of a ``peer gossip" protocol, the system remains scalable if the peers are in the same approximate geographical area which implies low average network latency.

\smallskip
\noindent \textit{Blockchain operations.} As discussed, aggregators' role is to aggregate sensor data into blockchain transactions. Assuming that aggregators perform no ``lossy" operations (such as averaging techniques), they would just package many collected sensor data along with the respective signatures into a transaction which in turn would be submitted to the ordering service. If we assume as in \cite{DBLP:journals/corr/abs-1801-10228} a block size of 2MB, we can estimate how much signed sensor data a block can hold. 
Given the discussion in Section \ref{measurements-sign-verif}, a Hyperledger block could hold (at most) about 15800 signed sensor data using our hash-based scheme vs. 12700 using ECDSA.

\smallskip
\noindent \textit{Latency.} We also wish to estimate the time from a value being proposed by an aggregator until consensus has been reached on it (assuming the block contains a single transaction). Again we can adopt previous evaluations in Hyperledger Fabric~\cite{DBLP:journals/corr/abs-1801-10228}, which show an average of 0.5 sec for the complete process.
Finally, considering that the previous evaluations mentioned above were all preformed on the original Hyperledger Fabric (while our architecture requires a slight modification as discussed in Section \ref{mod:hyperledger}), for our purposes we assume that the expected performance of aggregators (which are essentially Hyperledger peers also having client application functionalities) is not affected by this additional functionality, since the main affecting factor that can potentially become a bottleneck for the scalability of the whole system is network latency and not computational power.

\begin{table}[t]
\resizebox{\columnwidth}{!}{%
\begin{tabular}{l|c|c|c|c|}
	\cline{2-5}
	& \multicolumn{2}{c|}{\sysname}                            & \multicolumn{2}{c|}{ECDSA}                                              \\ \hline
	\multicolumn{1}{|l|}{Message length} & \multicolumn{1}{l|}{Sensor Sign} & \multicolumn{1}{l|}{Aggr Vrfy} & \multicolumn{1}{l|}{Sensor Sign} & \multicolumn{1}{l|}{Aggr Vrfy} \\ \hline
	\multicolumn{1}{|l|}{50}             & 50.43                            & 0.0339                               & \multirow{5}{*}{4200}            & \multirow{5}{*}{42.55}               \\ \cline{1-3}
	\multicolumn{1}{|l|}{100}            & 53.47                            & 0.0349                               &                                  &                                      \\ \cline{1-3}
	\multicolumn{1}{|l|}{150}            & 56.40                            & 0.0357                               &                                  &                                      \\ \cline{1-3}
	\multicolumn{1}{|l|}{202}            & 59.33                            & 0.03687                               &                                  &                                      \\ \cline{1-3}
	\multicolumn{1}{|l|}{218}            & 60.06                            & 0.0369                              &                                  &                                      \\ \hline
	\multicolumn{1}{|l|}{Signature size} & \multicolumn{2}{c|}{32}                                                 & \multicolumn{2}{c|}{64}                                                 \\ \hline
\end{tabular}
}
	\caption{Signing and verification costs (in milliseconds) compared with message and signature sizes (in bytes). Note we assume hash-based signatures are aggregated as discussed in Section \ref{measurements-sign-verif}. Signer is ATmega328P microcontroller and verifier is RPi 3.} 
\label{sign-verify-costs}

\end{table}

\begin{figure}[t]
	\centering
	\resizebox{\figurewidth\textwidth}{!}{
		\input{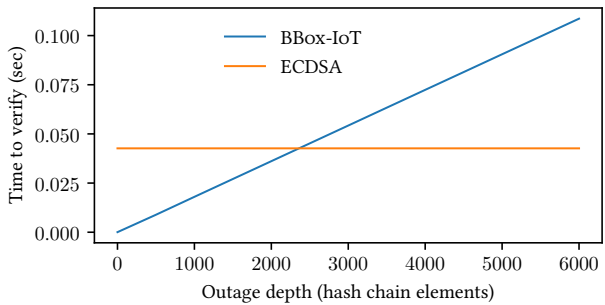}
	}
	\caption{Aggregator verification costs in network outages. \sysname is more expensive when more than about 2400 signature packets are lost.} 
	\label{hash-vs-ecdsa-outage}
\end{figure}

\definecolor{red2}{HTML}{D7191C}
\definecolor{orange2}{HTML}{FDAE61}
\definecolor{green2}{HTML}{ABDDA4}
\definecolor{blue2}{HTML}{2B83BA}
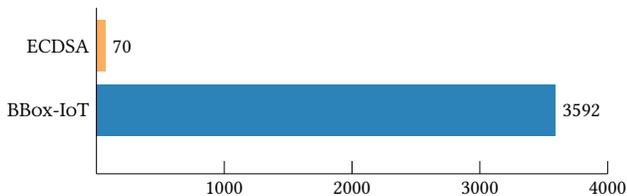
\begin{figure}[t]
	\centering
	\resizebox{\figurewidth\textwidth}{!}{
		\begin{tikzpicture}[x={(.002,0)}]
		\foreach  \l/\x/\c[count=\y] in {\sysname/3592/blue2, 
			ECDSA /70/orange2}
		{\node[left] at (0,\y) {\l};
			\fill[\c] (0,\y-.4) rectangle (\x,\y+.4);
			\node[right] at (\x, \y) {\x};}
		\draw (0,0) -- (4000,0);
		\foreach \x in {1000, 2000, ..., 4000}
		{\draw (\x,.2) -- (\x,0) node[below] {\x};}
		\draw (0,0) -- (0,2.6);
		\end{tikzpicture}
	}
	\caption{Number of signing operations for a 20mWh battery. 
	} 
	\label{energy-efficiency}
	\iffull
	\else
	\vspace{-0.1in}
	\fi
\end{figure}

\section{Related work}
\label{relwork}
We now discuss a number of works that connect IoT to the blockchain setting or works which build cryptographic primitives to optimize different parts of computation for resource-constrained IoT devices. Note that none of these works addresses the problem of authentication for extremely constrained (Class 0) devices.

\subsection{IoT and Blockchain}

 Shafagh et al.~\cite{Shafagh:2017:TBA:3140649.3140656} presented an architecture aiming to handle IoT data in a decentralized manner while achieving confidentiality, authenticity and integrity. This proposed system defines itself as ``IoT compatible" being append-only by a single writer and can be accessed by many readers, and consists of a layered design on top of an existing public blockchain to store access permissions and hash pointers for data, while storing the actual data off-chain using decentralized P2P storage techniques. Other approaches~\cite{8704309,8306880,8621042} also used a similar "layering" paradigm. While these approaches are simpler than ours, they ultimately rely heavily on the performance and properties of the underlying public blockchain and are not specifically tailored to handle resource-constrained IoT devices. 

Dorri, Kanhere, and  Jurdak~\cite{DBLP:journals/corr/DorriKJ16} considered a ``local" private blockchain maintained by a capable device, managed by the on-site owner and containing the local IoT device transactions. These lower-tier elements would be overlaid by a shared blockchain that can handle hashed data originating from the local blockchain and stored in a cloud storage service, and can enable access to local data. The above approach also offers confidentiality and integrity for submitted data and is suitable for resource-constrained IoT devices, however it is more complex than \sysname and requires managing and replicating data over several points in the system.

More recently, AlTawy and Gong~\cite{DBLP:journals/popets/AlTawyG19} presented a blockchain-based framework in the supply chain setting using RFIDs. This model considered blockchain smart contracts interacting with an overlay application on the RFID readers and a centralized server that handles membership credentials. This framework offers anonymity for the participating entities, which prove their membership in zero-knowledge, while their anonymity remains revocable by the server. It also provides confidentiality for its transactions and enforces a notion of ``forward secrecy" which enables future product owners in the supply chain to access its entire history. \sysname differs from the above work in several ways, since it is tailored to handle resource-constrained devices. Our work does not have confidentiality or anonymity as a main goal, although it can be added as an option using symmetric keys. We also do not require any smart contract functionality from the blockchain, and we operate exclusively in the permissioned setting.

IoTLogBlock \cite{DBLP:conf/lcn/ProfentzasAL19} shares a common goal with our work: enabling the participation of low-power devices in a distributed fashion, and similarly uses Hyperledger as a ``cloud service'' in a IoT setting. The crucial difference with our work, is that IoTLogBlock is evaluated on a Class 2 device using ECDSA signatures, which are far more expensive than our proposed hash-based signature and could not have been supported at all by a Class 0 device, while having much larger power consumption (Fig \ref{energy-efficiency}). Our proposed signature scheme is a key component for efficient implementations of blockchain-based systems in the IIoT setting.

Several more approaches have been presented which augmented an IoT infrastructure with a blockchain, focusing on providing two-factor authentication~\cite{8390280}, managing or improving communication among IoT devices~\cite{8029217,8378971}, implementing a trust management system in vehicular networks~\cite{8358773}, providing edge computing services~\cite{8436042}, data resiliency~\cite{8170858}, providing secure and private energy trade in a smart-grid environment~\cite{7589035} and implementing a hierarchical blockhain storage for efficient industrial IoT infrastructures \cite{DBLP:conf/blockchain2/WangSNH19} and 
all of which are orthogonal to our work. We point the reader to~\cite{DBLP:journals/comsur/AliVPDAR19,DBLP:journals/iotj/FerragDMDMJ19} for extensive reviews on the related literature. 

\iffull
\subsection{Hash-based Signatures}
\label{prel:onetime-sigs}
Early works such as Lamport's One-Time Signatures (OTS)~\cite{Lamport79} allowed the use of a hash function to construct a signature scheme. Apart from being one-time however, this scheme suffered from large key sizes. Utilizing tree-based structures such as Merkle trees \cite{C:Merkle87}, enabled to sign many times while keeping a constant-sized public key as the Merkle root. Winternitz OTS and later WOTS+~\cite{EPRINT:BDEHR11}\cite{AFRICACRYPT:Hulsing13} introduced a way of trading space for computation for the underlying OTS, by signing messages in groups. XMSS \cite{PQCRYPTO:BucDahHul11} further optimized the Merkle tree construction using Winternitz OTS as an underlying OTS. Other works such as HORS~\cite{ACISP:ReyRey02} enabled signing more than once, and more recently SPHINCS and SPHINCS+ ~\cite{EC:BHHLNP15,CCS:BHKNRS19} enabled signing without the need to track state.
Using HORS~\cite{ACISP:ReyRey02} as a primitive combined with a hash chain, Time Valid One-Time Signature (TV-HORS)~ \cite{DBLP:conf/infocom/WangKHN09} improves in signing and verification computational efficiency, but assuming ``loose'' time synchronization between the sender and the verifier.
All of the above scheme families while only involving hash-based operations, still incur either large computational and/or space costs, and cannot be implemented in Class 0 resource-constrained devices we consider. Follow-up work exists for implementing SPHINCS on resource-constrained devices~\cite{PKC:HulRijSch16} which we discuss later in this section and compare in Appendix \ref{apdx:modsphincs}.

The TESLA Broadcast Message Authentication Protocol~\cite{848446,Perrig02thetesla} follows a ``one-way'' chain-based approach for constructing a hash-based message authentication scheme. Based on a ``seed'' value, it generates a one-way chain of $n$ keys, which elements are used to generate temporal MAC keys for specified time intervals. The protocol then discloses each chain element with some time delay $\Delta$, then authenticity can be determined based on the validity of the element in the chain as well as the disclosure time.
The ``pebbling'' algorithms~ \cite{jakobsson2002fractal,RSA:YSEL09} enable logarithmic storage and computational costs as discussed in Section \ref{our-primitive}.
Its main drawback however is that it also requires ``loose'' time synchronization between the sender and the receiver for distinguishing valid keys. In an IoT setting this would require the frequent execution of an interactive synchronization protocol, since IoT devices are prone to clock drifting \cite{DBLP:journals/sensors/Tirado-AndresRA19,DBLP:journals/tii/ElstsFDOPC18}. Also we assume in Section \ref{properties} that IoT devices function in a broadcast-only mode, which would not allow the execution of such interactive protocol in the first place. Furthermore, TESLA introduces a ``key disclosure delay'' which might be problematic in certain IoT applications, and gives up the non-repudiation property of digital signatures.

Several modifications and upgrades to the TESLA protocol have been proposed, with most of them maintaining its ``key disclosure delay'' approach which is also associated with the loose time synchronization requirement \cite{NDSS:LiuNin03,ACNS:HuJakPer05}. A notable paradigm is the ``hierarchical'' (or two-dimensional) one-way chain structure, where the elements of a ``primary'' hash chain serve as seeds for ``secondary'' chains in order to reduce communication costs. \cite{ACNS:HuJakPer05} includes several such proposals.
For instance, its Sandwich-chain uses two separate one-way chains. The first one-way chain is used as a ``primary'' chain, which generates intermediate ``secondary'' chains using the elements of the second one-way chain as salts. However to maintain efficiency, it still assumes some weak time synchronicity between the signer and the verifier by disclosing each element of the ``primary'' chain with some time delay (else the verifier in case of a network outage would have to recompute all the previous secondary chains as well which would defeat its efficiency gains). More importantly however, this construction has much larger storage requirements than ours.
In the same work, the Comb Skipchain construction is asymptotically more efficient in signing costs than our scheme and does not require time synchronicity, but has worse concrete storage requirements which are prohibitive for low-end IoT devices, and still suffers from delayed  verification. This work includes other interesting modifications such as the  ``light'' chains where the secondary chains are generated using a lower security parameter
and a standard one-dimensional TESLA variant which does not require a MAC.

\else

\subsection{Hash-based Signatures}
\label{prel:onetime-sigs}
Lamport's One-Time Signatures (OTS)~\cite{Lamport79} was the first scheme to allow the use of a hash function to construct a signature scheme. Then, Winternitz OTS and  WOTS+~\cite{EPRINT:BDEHR11}\cite{AFRICACRYPT:Hulsing13}enabled a time-memory tradeoff by signing messages in groups, used in turn by XMSS \cite{PQCRYPTO:BucDahHul11} in a Merkle tree construction. Other works such as HORS~\cite{ACISP:ReyRey02} enabled signing more than once, and more recently SPHINCS and SPHINCS+ ~\cite{EC:BHHLNP15,CCS:BHKNRS19} enabled signing without the need to track state.
Using HORS~\cite{ACISP:ReyRey02} as a primitive combined with a hash chain, Time Valid One-Time Signature (TV-HORS)~ \cite{DBLP:conf/infocom/WangKHN09} improves in signing and verification computational efficiency, but assuming ``loose'' time synchronization between the sender and the verifier.All of the aforementioned schemes, while only involving hash-based operations, still incur large computational and/or space costs and cannot be implemented in Class 0 resource-constrained devices we consider.

TESLA~\cite{848446,Perrig02thetesla} constructs a ``one-way'' hash chain to generate temporal MAC keys for specified time intervals, disclosing each chain element with some time delay $\Delta$.
While ``pebbling'' algorithms~ \cite{jakobsson2002fractal,RSA:YSEL09} enable logarithmic storage and computational costs as discussed in Section \ref{our-primitive}, it requires ``loose'' time synchronization between the sender and the receiver for distinguishing valid keys. In an IIoT setting this would require the frequent execution of an interactive synchronization protocol since such devices are prone to clock drifting \cite{DBLP:journals/sensors/Tirado-AndresRA19,DBLP:journals/tii/ElstsFDOPC18}.
Several modifications and upgrades to TESLA have been proposed, but most of them still require time synchronization \cite{NDSS:LiuNin03,ACNS:HuJakPer05}. 
\fi

\subsection{Cryptographic Operations in IoT}  In the context of improving cryptographic operations in the IoT setting, Ozmen and Yavuz~\cite{DBLP:conf/ccs/OzmenY17} focused on optimizing public key cryptography for resource-constrained devices. This work exploited techniques in Elliptic Curve scalar multiplication optimized for such devices and presented practical evaluations of their scheme on a low-end device. Even though the device used in this work is can be classified as a Class 1 or Class 2 device, our construction  signing is more efficient both in terms of computation cost and storage by at least an order of magnitude.

\iffull As discussed above, \fi H{\"u}lsing, Rijneveld and Schwabe \cite{PKC:HulRijSch16} showed a practical evaluation of the SPHINCS hash-based signature scheme \cite{EC:BHHLNP15} on a Class 2 device. At first glance this implementation could also serve our purposes, however our proposed construction, while stateful, is much cheaper in terms of runtime, storage and communication costs, without such additional assumptions. \iffull We directly compare with their scheme  in Appendix \ref{apdx:modsphincs}.\fi

Kumar et al.~\cite{DBLP:journals/corr/abs-1905-13369} propose an integrated confidentiality and integrity solution for large-scale IoT systems, which relies on an identity-based encryption scheme that can distribute keys in a hierarchical manner. This solution also uses similar techniques to our work for signature optimization for resource-constrained devices, however, it requires synchronicity between the system participants. Portunes~\cite{DBLP:conf/smartgridcomm/LiDN14} is tailored for preserving privacy (which is not within our main goals in our setting), and requires multiple rounds of communication (while we consider a “broadcast-only” setting)

\iffull
Wander et al.~\cite{Wander:2005:EAP:1048930.1049786} quantified the energy costs of RSA and Elliptic Curve operations as public key cryptography algorithms in resource-constrained devices. In a similar context, Potlapally et al.~\cite{1231830} performed a comprehensive analysis of several cryptographic algorithms for battery-powered embedded systems. However as discussed in Section  \ref{prel:onetime-sigs}, we consider hash-based algorithms that are lighter and more efficient.\fi

Finally we mention an extensive IoT authentication survey \cite{DBLP:journals/sensors/El-hajjFCS19}. In this work, our authentication scheme is comparable to \cite{DBLP:conf/esweek/BamasagY15} which utilizes hashing for one-way authentication in a distributed architecture, however our scheme is more storage-efficient, suited for low-SWaP (Class 0) sensors.

\section{Conclusions}
\label{conclusions}
In this paper we designed and implemented \sysname, a block-chain
inspired approach for Industrial IoT sensors aiming at offering a
transparent and immutable system for sensing and control information
exchanged between IIoT sensors and aggregators. Our approach
guarantees blockchain-derived properties to even low-Size Weight and
Power (SWaP) devices. Moreover, \sysname acts as a "black-box" that
empowers the operators of any IoT system to detect data and sensor
tampering ferreting out attacks against even SWaP devices. We posit
that enabling data auditing and security at the lowest sensing level
will be highly beneficial to critical infrastructure environments with
sensors from multiple vendors.

Finally, we envision that our approach will be implemented during the
sensor manufacturing stage: having industrial sensors shipped with
pre-computed pebbles and their key material labeled using QR-code on
the sensor body will allow for a seamless and practical deployment of
\sysname.

\iffull
\begin{acks}
	Foteini Baldimtsi and Panagiotis Chatzigiannis were supported by NSA 204761.
\end{acks}
\fi

\bibliographystyle{ACM-Reference-Format}
\bibliography{mybibliography,abbrev3,crypto,iot}

\iffull

\appendix
\section{On MACs for sensor authentication}
\label{constr:notmacs}
One might suggest using MAC authentication in our scheme instead of one-time hash based signatures, which might be slightly more efficient in terms of computation cost for generating a signature, are simpler in usage and do not expire. The question of whether its preferable using symmetric primitives in resource-constrained IoT devices instead of public key cryptography has been raised in academic works~\cite{DBLP:conf/ccs/OzmenY17}, and several motivations to provide efficient public key cryptography techniques in such devices were outlined, most of which are also applicable to our system as follows.

Firstly, signatures provide non-repudiation, which as discussed previously is a needed security property \ref{sec-nonrepud}. Although a way to achieve non-repudiation through MACs could be to use a separate MAC key for each sensor, each key would need to be shared with each group aggregator separately since they should be all able to verify data from all sensors in the group. This would increase the attack surface since an attacker compromising any aggregator could also send bogus data for all sensors. Also considering that aggregators might have to verify data from a great number of sensors, our hash-based verification cost (which involves one hash operation) is cheaper than one MAC operation. Although for sensors a MAC operation is cheaper than a hash-based signature, as we show in section \ref{measurements} a hash-based signature which involves a few hashes and a Quicksort operation is still relatively efficient even for the weakest types of sensors.

Secondly, our chain  hash-based scheme has a built-in ``replay protection" against an attacker, since that signature is by definition valid for one time only. A MAC scheme would require extra layers of protection (nonces and/or timers) against replay attacks.

Lastly, by using our  hash-based signature scheme we enable public verifiability of signed sensor data on the blockchain, even by entities not authorized to participate in the system.

\section{Chain-Based Hash Signatures}
\label{apdx:chainSign}

\subsection{Digital Signatures.}
A digital signature scheme consists of the following algorithms \cite{Katz:2014:IMC:2700550}:
\begin{itemize}
	\item $\signgen{\publickey{}{}}{\secretkey{}{}}$: Outputs a pair of keys $(\publickey{}{},\secretkey{}{})$.
	\item $\sign{\secretkey{}{}}{m}{\sigma}$: Takes as input a private key $\secretkey{}{}$ and a message $m$ and outputs a signature $\sigma$.
	\item $\svrfy{\publickey{}{}}{m}{\sigma}$: Takes as input a public key $\publickey{}{}$, a message $m$ and a signature $\sigma$, and outputs a bit $b$ where $b=1$ indicates successful verification.
\end{itemize}
A digital signature is considered secure if an adversary $\adv$ cannot forge a signature on a message even after adaptively receiving signatures on messages of its choice. To formalize the security definition we first describe the following experiment $\mathsf{SigForge(\secpar)}$:
\begin{enumerate}
	\item $\signgen{\publickey{}{}}{\secretkey{}{}}$
	\item $\adv$ on input $(\publickey{}{})$ queries the signing oracle polynomial number of times $q$. Let $Q: [m_{i},\sigma_{i}]_{i=1}^{q}$ be the set of all such queries.
	\item $\adv$ outputs $(m^{*},\sigma^{*})$.
	\item $\adv$ wins if $\mathsf{SVrfy}(m^{*},\sigma^{*})$ = 1 where $m^{*} \notin [m_{i}]_{i=1}^{q}$ and $\mathsf{SigForge}$ outputs ``1", else it outputs ``0".
\end{enumerate}
\begin{defn}
	A digital signature scheme is existentially unforgeable under an adaptive chosen-message attack, if for all PPT $\adv$, $\pr{\mathsf{SigForge(\secpar) = 1}}$ is negligible in $\secpar$.
\end{defn}

\subsection{One-time signatures} A digital signature scheme that can be used to sign only one message per key pair is called a one-time signature (OTS) scheme.

\begin{defn}
	\label{prel:onetime-sigs-defn}
	A one-time digital signature scheme is existentially unforgeable under an adaptive chosen-message attack, if for all PPT $\adv$ and for $q\leq 1$, $\pr{\mathsf{SigForge(\secpar) = 1}}$ is negligible in $\secpar$.
\end{defn}

\subsection{Hash functions.}
An (unkeyed) hash function  $y:=h(m)$ on input of a message $m$ outputs a digest $y$.  
A cryptographic hash function is considered secure if the probability to find collisions is negligible (i.e. it is \emph{collision resistant}). More formally, we consider the following experiment $\mathsf{HashColl}$\cite{Katz:2014:IMC:2700550}:
\begin{enumerate}
	\item $\adv$ picks values $x, x' \in \{0,1\}^{*}$ s.t. $x \neq x'$.
	\item $\adv$ wins if $h(x) = h(x')$ and $\mathsf{HashColl}$ outputs ``1".
\end{enumerate}

\begin{defn}
	\label{prel:hash-secdef}
	A hash function $h(): \{0,1\}^* \rightarrow \{0,1\}^\secpar$ is collision resistant if for all PPT $\adv$, $\pr{\mathsf{HashColl = 1}}$ is negligible.
\end{defn}

A weaker notion for security of a hash function is preimage-resistance. We consider the following experiment $\mathsf{PreIm}(\secpar,y)$:
\begin{enumerate}
	\item $\adv$ is given $y \in \{0,1\}^\secpar$ 
	\item $\adv$ outputs $x$.
	\item $\adv$ wins if $h(x)=y$. If $\adv$ wins $\mathsf{PreIm}$ outputs ``1", else it outputs ``0".
\end{enumerate}
\begin{defn}
	\label{prel:hash-preim}
	A hash function $h(): \{0,1\}^* \rightarrow \{0,1\}^\secpar$ is preimage resistant if \space $\forall$ ppt $\adv$ and $\forall y \in \{0,1\}^\secpar$,
	$\pr{\mathsf{PreIm}(\secpar,y) = 1}$ is negligible in $\secpar$.
\end{defn}
\begin{corr}
	A collision resistant hash function is also preimage resistant.
\end{corr}

\subsection{Definition and Security proof}
\label{apdx:proofs}

We first define the API of a chain based signature for a fixed number of messages $n$.
\begin{itemize}
	\item $\otkeygen{\publickey{}{}}{\secretkey{n}{}}{s_{0}}{n}$: Outputs a pair of keys $\publickey{}{},\secretkey{n}{}$ and an initial state $s_{0}$, where $\publickey{}{} = h^{n}(\secretkey{n}{})$ and $h()$ is a collision resistant hash function. 
	\item $\otsign{\secretkey{i}{}}{\secretkey{i-1}{}}{s_{i}}{s_{i-1}}{m}{\sigma}$: Takes as input the system state $s_{i-1}$, a private key  $\secretkey{i-1}{}$ and a message $m$, generates a signature $\sigma$ and updates the signer's private key to $\secretkey{i}{}$ where $\secretkey{i-1}{} = h(\secretkey{i}{})$ and his state to $s_{i}$ where $i \leq n$.
	\item $\otverify{\publickey{}{}}{m}{\sigma}$: Takes as input a public key $\publickey{}{}$, a message $m$ and a signature $\sigma$, and outputs a bit $b$ where $b=1$ indicates successful verification. 
\end{itemize}

To formalize security for chain-based signatures with length of chain $n$, we describe the following experiment $\mathsf{OTSigForge}(\secpar,n)$:
\begin{enumerate}
	\item $\otkeygen{\publickey{}{}}{\secretkey{n}{}}{s_{0}}{n}$
	\item $\adv$ on input $(\publickey{}{},n)$ makes up to $q \leq n$ queries to the signing oracle. Let $Q: [m_{i},\sigma_{i}]_{i=1}^{q}$ the set of all such queries where $m_{i}$ is the queried message and $\sigma_{i}$ is the signature returned for $m_{i}$.
	\item $\adv$ outputs $(m_{q+1},\sigma_{q+1})$.
	\item $\adv$ wins if $\mathsf{OTVerify}(\publickey{}{},m_{q+1},\sigma_{q+1}):= 1$ and  $h^{i}(\secretkey{i}{}) \neq \publickey{}{}$ $\forall i \leq q$ where $\mathsf{OTSigForge}$ outputs ``1", else it outputs "0".
\end{enumerate}

Note in the above experiment by $h^{i}(\secretkey{i}{}) \neq \publickey{}{}$ $\forall i \leq q$ we restrict $\adv$ from winning the game by reusing a secret key $\secretkey{i}{}$ existing in the chain up to distance $q$ from the public key $\publickey{}{}$.

\begin{defn}
	\label{defn:OTS}
	A chain-based one-time digital signature scheme is existentially unforgeable under an adaptive chosen-message attack, if \space $\forall$ ppt $\adv$, $\pr{\mathsf{OTSigForge}(\secpar,n) = 1}$ is negligible in $\secpar$.
\end{defn}

Given the formal definition above we now prove the security of Construction \ref{prel:onetime-sigs-constr}.

\begin{thm}
	Let $h: \{0,1\}^{*} \rightarrow \{0,1\}^{\secpar}$ be a  preimage resistant hash function. 
	Then Construction \ref{prel:onetime-sigs-constr} is an existentially unforgeable chain-based one-time signature scheme. 
\end{thm}
\begin{proof}
	Let $\adv$ be an adversary who wins the $\mathsf{OTSigForge}$ game described in Section \ref{prel:onetime-sigs} and therefore can forge signatures using the above scheme with non-negligible probability $p(\secpar)$. That is, $\exists \ \adv$ which after performing $q$ queries $\{m_{i}\}_{i=1}^{q}$ where $q \leq n$, can output a signature $\sigma_{q+1}$ for a message $m_{q+1}$ where $\mathsf{OTVerify}(\publickey{}{},m_{q+1},\sigma_{q+1})$ = 1 and $h^{i}(\secretkey{i}{}) \neq \publickey{}{}$ $\forall i \leq q$. 
	
	Then, an algorithm $\mathcal{B}$ running the $\mathsf{PreIm}$ experiment would use $\adv$ to break preimage resistance of $h$ as follows: On input $(\secpar,y)$, $\mathcal{B}$ would generate a hash chain of length $n$ with seed $y$ as $(y, h(y), ... ,h^{n}(y))$ and forward  $(h^{n}(y),n)$ to $\adv$. 
	Then $\adv$ makes up to $q \leq n$ queries to $\mathcal{B}$. When $\adv$ queries for $m_{i}$ (where $i$ denotes the query number), $\mathcal{B}$ returns $\sigma_{i} = h(m_{i}||h^{n-i+1}(y))||h^{n-i}(y)$ to $\adv$. If $q = n$ and $\adv$ does not output a forgery, $\mathcal{B}$ returns $\bot$ and starts over.  
	If $\adv$ eventually outputs a forgery $(m^{q+1},\sigma^{q+1})$ to $\mathcal{B}$ and $q < n$,  then $\mathcal{B}$ returns $\bot$ as output of $\mathsf{PreIm}$ experiment and starts over, else if $q=n$, $\mathcal{B}$ would parse $\sigma^{n+1} = \sigma^{A}||\sigma^{B}$ and return $\sigma^{B}$. Assuming a uniform probability distribution of the number of queries $q$, $\pr{\mathsf{PreIm}(\secpar,y) = 1} = \frac{\pr{\mathsf{OTSigForge}(\secpar,n) = 1}}{n}  = \frac{p(\secpar)}{n}$ which is non-negligible.
	
\end{proof}

\subsection{Evaluation comparison with modified SPHINCS}
\label{apdx:modsphincs}
As discussed in section \ref{prel:onetime-sigs}, the modified SPHINCS scheme tailored for resource-constrained devices~\cite{PKC:HulRijSch16} could be a candidate scheme for our system. Here we make a direct comparison between modified SPHINCS and our proposed scheme for our system's purposes. 

Assuming a hash chain length of $2^{26}$ elements (which as discussed is only exhausted after 21 years assuming generation of one signature every 10 seconds), a signature generation only requires 27 hashing operations in the worst case, which according to our measurements on a 8-bit 16Mhz CPU Arduino device outlined in Section \ref{measurements-sign-verif}, would only need 50 ms on average. 
On the other hand, modified SPHINCS' evaluation performed on a resource-constrained device (32-bit 32Mhz Cortex M3 which is more powerful than our Arduino Uno R3) needs 22.81 seconds for signature generation. Also our signature size (excluding the payload) is only 64 bytes for the signature and the program storage requirement 1082 bytes, while modified SPHINCS generates a 41KB signature, streamed serially. Our only additional requirement is an initial precomputation phase using a powerful device, which will have to pre-compute the $2^{26}$ hash chain elements and then send the ``pebbles" to the resource-constrained device.

\subsection{Collision probability analysis}
\label{measurements-coll-prob}

Although we assumed a collision resistant hash function in our hash-based signature construction, given the length of the hash chain (typical length $2^{26}$) there is an increased likelihood of a collision along that chain through the birthday paradox (especially for lower levels of security where the output size of the hash function is small), which would result in ``cycles" of hashes. If such cycles occur, an adversary could then trivially break the security of our scheme and sign bogus sensor data.

Assume a hash chain of length $2^{n}$ and a security parameter $\secpar$. From the birthday paradox, the probability of a collision on the hash chain is approximated by $p = 1 - e^{\frac{-2^{n}(2^{n}-1)}{2^{\secpar+1}}}$. In Figure \ref{coll_prob_secpar} we show that given a chain length of $2^{26}$ as previously discussed, the output size of the hash function $h()$ should be at least 64, and SHA256 which we used in our evaluations satisfies these requirements.

Nevertheless, if birthday attacks become an issue for a small security parameter, we can apply the technique from  \cite{ACNS:HuJakPer05} where the chain index is used as salt to prevent such attacks for a small overhead in cost. However since we show that the birthday attack is negligible, we prefer to keep the costs as low as possible.

\begin{figure}[]
	\centering
	\resizebox{\figurewidth\textwidth}{!}{
		\input{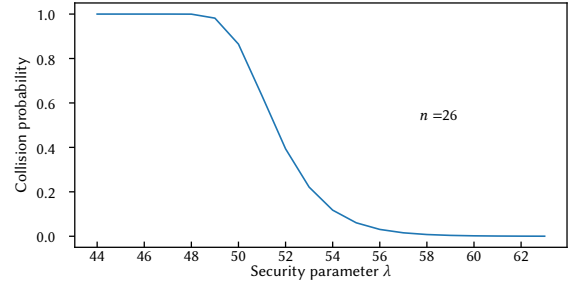}
	}
	\caption{Collision probability for hash chain length $2^{26}$.} 
	\label{coll_prob_secpar}
\end{figure}

\section{Consensus}
\label{apdx:consensus}
\subsection{Definitions}

\begin{defn}
	\label{prel:cons_def}
	Let parties $[P_{i}]_{i=1}^{n}$, each having a view of the blockchain $\blockchain(i)$, and receive a common sequence of messages in rounds $[1..j..]$
	A protocol solves the ledger consensus problem if the following properties hold:
	\begin{enumerate}[label=\roman*.]
		\item Consistency: An honest node's view of the blockchain on some round $j$ is a prefix of an honest node's view of the blockchain on some round $j+\ell, \ell>0$, or $\blockchain(i)_{j}||\block{j+1}||...||\block{j+\ell} = \allowbreak\blockchain(i')_{j+\ell},\allowbreak \forall P_{i},P_{i'},j,\ell$.
		\item Liveness: An honest party $P_{i}$ on input of an operation (or transaction) $\mathsf{tx}$, after a certain number of rounds  will output a view of the blockchain $\blockchain(i)$ that includes $\mathsf{tx}$.
	\end{enumerate}
\end{defn}

The protocol can be augmented with the existence of a $\trustedparty{}$ assigning membership credentials (which requires an additional trusted setup phase) resulting in an \emph{Authenticated} ledger consensus protocol~\cite{Lamport:1982:BGP:357172.357176}\cite{EPRINT:LinLysRab04}. Such a protocol consists of the following algorithms:

\begin{enumerate}
	\item $(\param_{}) \leftarrow \mathsf{TPSetup}(1^{\secpar})$: A trusted party $\trustedparty{}$ generates the system paremeters $\param_{}$.
	\item $(\publickey{i}{},\secretkey{i}{}) \leftarrow \mathsf{PartyGen}(\param_{})$: Each party $i$ generates a public-private key pair.
	\item $\mathsf{TPMembers}(\widehat{[\publickey{}{}]}) := [\publickey{}{}]$: The $\trustedparty{}$ chooses the protocol participants from a list of public keys $\widehat{[\publickey{}{}]}$, and outputs a list of authenticated participant public keys $[\publickey{}{}]$.
	\item $\consensus ([[\publickey{j}{}]_{j=1}^{n},\secretkey{i}{},\nstate{i},\blockchain(i)]_{i=1}^{n}) := \blockchain'$: All system participants with state $\nstate{i}$ and a view of the blockchain $\blockchain(i)]_{i=1}^{n})$, agree on a new blockchain $\blockchain'$.
\end{enumerate}

\subsection{Byzantine Generals Problem}
One of the first studies attempting to achieve agreement on a distributed synchronous system was proposed by Leslie Lamport in 1982 \cite{Lamport:1982:BGP:357172.357176}. It proposed using either an ``Oral message" solution to achieve binary consensus among $n = 3f + 1$ nodes, where a leader was sending a proposed binary value in a fully connected network, then the honest nodes would propagate that value using the same protocol acting as leaders themselves. Using authentication further improves the resilience of the protocol to $n = 2f + 1$, as Byzantine faults can be tolerated under the assumption that dishonest parties cannot forge the leader's signature. In both cases however, the communication complexity is $O(n^{2})$, which is not considered scalable.

\subsection{PBFT}
The Practical Byzantine Fault Tolerance consensus protocol \cite{Castro:1999:PBF:296806.296824} was an exemplary protocol for many consensus protocols to follow. It assumes a partially asynchronous setting (i.e. offering \emph{eventual synchrony}), where a message is assumed to arrive to the destination node in the network after some unknown but bounded time $t$. In each round, a leader orders messages and propagates them to all nodes in the network in three phases. The leader is defined by a sequence of ``views", and the backup nodes can propose a leader (or view) change if he has faulty behavior (timeout exceeded). The protocol assumes $n = 3f + 1$ faulty nodes, and has communication complexity $O(n^{2})$, which again is in practice not scalable for more than 20 nodes \cite{DBLP:conf/ifip11-4/Vukolic15}. It also assumes static membership of participating nodes. 

\subsection{PBFT with dynamic clients}
Chondros et al. \cite{DBLP:conf/middleware/ChondrosKR12} suggest modifying the original PBFT protocol by adding ``Join" and ``Leave" system requests. The ``Join" request, equivalent to a ``client" request in PBFT, would trigger the execution of the original PBFT protocol between the primary and the replicas, and replying back with a challenge to prevent DoS attacks. The client wishing to join, replies with a message containing its credentials and a nonce, which are then associated in an array containing the protocol participants. The protocol is based on fully synchronous assumptions. While the protocol offers dynamic membership, it has high communication overhead as PBFT, therefore it is not a suitable candidate for our \sysname system.

\subsection{System membership tracking consensus}
Rodrigues et al. \cite{DBLP:journals/tdsc/RodriguesLCLS12} propose a ``Membership service" or a ``Configuration management service" for achieving dynamic membership in the permissioned consensus setting. Following a ``transaction endorsement - consensus nodes decoupling" paradigm, only a chosen subset of the participating nodes would execute a BFT-family protocol, while also being responsible for handling ``add" or ``remove" requests, pinging (or probing) for failed or crashed nodes and removing them after some time (or epochs), and keeping track of the epochs. Since these are only a small subset of the total nodes, the overall system remains scalable. Adding/removal is signed by a TP, which is compatible with our \sysname architecture.

\subsection{RSCoin}
Designed as a cryptocurrency framework, RSCoin~\cite{NDSS:DanMei16} deviates from the common practice of decentralizing monetary supply, which found in other cryptocurrencies. In RSCoin's architecture, a Trusted Party (the Bank) delegates its authority for validating transactions to known and semi-trusted ``mintettes", which in turn are each responsible for interacting with a subset of the cryptocurrency's addresses, forming ``shards". A basic consensus protocol based on Two-Phase Commit is executed between a group of mintettes and a client, which involves collecting signatures from the majority of mintettes responsible and then sending back the transaction to be included in a block. 
Also, since there is no direct communication between all mintettes, but they rather communicate indirectly through the clients, the communications complexity is very low which enables high scalability and performance. RSCoin could also potentially be used as a ``consensus'' replacement for \sysname, however it would require extensive modifications  because of its cryptocurrency-oriented architecture. 

\subsection{Hyperledger Frameworks}
In the following paragraphs we summarize the properties of additional Hyperledger frameworks \cite{Hyperledger-architecture-vol1}:
\begin{itemize}

	\item \textbf{Hyperledger Indy} uses the \textit{Redundant Byzantine Fault Tolerance (RBFT)} consensus algorithm \cite{6681599}, which is voting-based. As its name suggests, it is based on the PBFT consensus algorithm, modified to execute many protocol instances in parallel for improving performance in the case of a faulty leader. It provides Byzantine fault tolerance and reaches consensus very fast, however the time scales with the number of nodes ($O(n^{2})$ as in PBFT). An additional requirement is that the nodes must be totally connected.
	\item \textbf{Hyperledger Iroha} uses the \textit{Sumeragi} consensus algorithm, based on a reputation system. As with RBFT, it provides Byzantine fault tolerance reaching consensus in very short time, but that time scales with the number of nodes, and the nodes must be totally connected.
	\item \textbf{Hyperledger Sawtooth} uses the novel \textit{Proof of Elapsed Time (PoET)} consensus algorithm, which is lottery-based. It can be categorized in the ``Proof-of-X" consensus family, by replacing proof of computational work with a proof of elapsed time, using trusted hardware (Intel SGX enclave). Each protocol participant would request a wait time from their trusted hardware, and the shortest would be elected as the leader, by providing a proof that he indeed had the shortest wait time and the new block alongside the proof was not broadcasted until that time had expired \cite{DBLP:journals/corr/abs-1711-03936}. While the algorithm is scalable and Byzantine fault tolerant, there is a possibility of delayed consensus due to forks. Also this algorithm would not suit to our \sysname system, since we do not require our blockchain ``maintainers" to have trusted hardware capabilities.
\end{itemize}

\subsection{Additional Consensus properties}
\label{apdx:consensus-properties}
\begin{table*}[]
	\begin{tabular}{|p{27mm}|p{17mm}|p{17mm}|p{18mm}|p{17mm}|p{17mm}|p{30mm}|}
		\hline
		Algorithm & Adversarial model & Byzantine tolerant & Dynamic membership & Scalable & DoS resistant & Notes\\
		\hline
		PBFT & $3f+1$ & \cmark & \xmark & \xmark & Semi &  Hyperledger Fabric v0.6\\
		Kafka & $2f+1$ & \xmark & \cmark & \cmark & \xmark & Hyperledger Fabric v1.4 \\
		BFT-SMaRt & $3f+1$ & \cmark & \cmark & Semi &\cmark &  \\
		Nakamoto consensus & $2f+1$ & \cmark & Permissionless & \cmark & \cmark & \\
		\hline
	\end{tabular}
	\caption{Consensus algorithm comparison}
	\label{prelims:consensus-comparison}
\end{table*}

We outline the following additional consensus properties, which are desirable in our setting but not strictly required: 
\begin{enumerate}[label=\roman*.]
	\item Byzantine Fault Tolerant: As previously discussed, the crash tolerant model of consensus does not take  Byzantine behavior of nodes into account~\cite{DBLP:journals/corr/CachinV17}. Although our system considers the consensus algorithm running among a closed, controlled set of nodes, it might be depoyable in a more uncontrolled environment, where Byzantine behavior is possible (Byzantine consensus)\cite{Lamport:1982:BGP:357172.357176}.
	\item No synchronicity assumptions: Due to \cite{DBLP:journals/jacm/FischerLP85}, which excludes deterministic protocols from reaching consensus in a fully asynchronous system, to achieve synchronicity, we would have to use a randomized protocol, else we will have to assume ``eventual synchrony" (i.e. protocol finishes within a fixed but unknown time bound) \cite{CCS:MXCSS16}.
	\item Incentive-compatible: Protocol keeps nodes motivated to participate in the system and follow its rules \cite{DBLP:journals/corr/abs-1711-03936}.
	\item Minimal setup assumptions: No need for a trusted setup phase (as required by Authenticated Byzantine agreement discussed above).
	\item Weaker adversarial model: While most classical consensus protocols require a 33\% adversarial model (which we believe should be sufficient for our purposes), some protocols have a weaker requirement of 49\% adversarial power. However this usually comes at the cost of sacrificing Byzantine tolerance (as we discussed above) or scalability in terms of operations per second \cite{Nakamoto:bitcoin}. 
\end{enumerate}

\subsection{An instantiation for Consensus algorithm}
\label{constr-pluggedconsens}
In the generic construction of our scheme, we assumed a ``pluggable" consensus algorithm, decoupled from our construction, similar to the original Hyperledger architecture. Recall that this algorithm, which is executed among all orderers $\orderer{i}$, on input of a blockchain $\blockchain$ and some orderer state $\nstate{i}$, outputs an agreed new updated blockchain $\blockchain'$. Here we provide a concrete instantiation of a consensus algorithm for the modified Hyperledger used in \sysname that matches the PBFT consensus protocol~\cite{Castro:1999:PBF:296806.296824} as follows (note though that PBFT would not satisfy all of the required system properties as discussed in section \ref{prel:consensus}): 
\begin{enumerate}
	\item $\orderer{i}$ parses $\orderertxlist$ from its $\nstate{\orderer{i}}$ extracting a set of transactions $\{\tx_{i}\}$.
	\item $\orderer{i}$ based on the current $\blockchain$ and $\{\tx_{i}\}$ constructs a new block $\block{i}$ which would create $  \blockchain ||\block{i} \rightarrow \blockchain'$.
	\item $\orderer{i}$ computes $\sigma := \mathsf{Sign}(\secretkey{\orderer{i}}{} , \block{i})$. and sends $\sigma$ to all orderers in $\ordererset$ (equivalent to ``pre-prepare" phase in PBFT).
	\item All the other orderers $\orderer{x} \in \ordererset$ parse $(\ordererlist_{\blockchain})$ from the output of $\readconfig (\blockchain)$. Check if $\publickey{\orderer{i}}{} \in \ordererlist_{\blockchain}$  and $\mathsf{SVrfy}(\publickey{\orderer{i}}{},\sigma,\block{i})=1$. Then it verifies that the proposed block was formed correctly (i.e., it is a valid extension of the current blockchain $\blockchain$). If all verifications pass, it computes $\sigma_{x} := \mathsf{Sign}(\secretkey{\orderer{x}}{} , \block{i})$ and sends $\sigma_{x}$ to all orderers in $\ordererset$ (equivalent to ``prepare" phase in PBFT).
	\item Each $\orderer{x} \in \ordererset$ (including $\orderer{i}$) checks if $\mathsf{SVrfy}(\publickey{\orderer{x}}{},\sigma_{x},\block{i})=1$. If it collects sufficient number of signatures (specific to each consensus protocol) it computes  $\sigma_{x}' := \mathsf{Sign}(\secretkey{\orderer{x}}{} , \block{i},1)$
	and sends $\sigma_{x}'$ to all orderers in $\ordererset$ (equivalent to PBFT ``commit" phase ). 
	\item Each $\orderer{x} \in \ordererset$ (including $\orderer{i}$) checks if \\ $\mathsf{SVrfy}(\publickey{\orderer{x''}}{},\sigma_{x}',\block{i},1)=1$. If it receives sufficient number of signatures (specific to each consensus protocol) it updates it state to $\nstate{\orderer{x}}'$ and outputs ``1". It outputs ``0" in all other cases.	
\end{enumerate} 
The above instantiation satisfied the basic consensus properties in Definition \ref{prel:cons_def}.

\begin{table*}[t]\centering
	\caption{Hash-based schemes concrete comparison, 256-bit security}
	\label{ots_concrete_comparison}
	\begin{tabular}{|l|l|l|l|l|l|l|p{25mm}|}
		\hline
		Scheme     & Stateful & Public key (bytes) & Secret key (bytes) & Signature (bytes) & Sign (msec) & Verify (msec) & Remarks \\ \hline
		XMSS       &  Yes        &    68        &            & 4963          &610      &    160    &  Cortex M3 32MHz 32-bit 16KB RAM   \cite{EPRINT:KamFlu17,EPRINT:HulRauBuc17}    \\ \hline
		SPHINCS    & No         &1056            &1088            & 41000         &  18410  &  513 & Cortex M3 32MHz 32-bit 16KB RAM  \cite{PKC:HulRijSch16}   \\ \hline
		Our scheme &  Yes        &  32          &   32         &   32 (64)        &  52    &   0.035      & ATmega328P 16MHz 8-bit 2KB RAM       \\ \hline
	\end{tabular}
\end{table*}

\section{Construction algorithms}
\label{sec-model-defs}

For our construction we assume an existentially unforgeable signature scheme $(\mathsf{SignGen},\mathsf{Sign},\mathsf{SVrfy})$
and an unforgeable one-time chain based signature scheme as defined in Section \ref{prel:onetime-sigs}  $(\mathsf{OTKeyGen}$, $\mathsf{OTSign}$, $\mathsf{OTVerify})$. We also assume an authenticated blockchain consensus scheme $(\mathsf{TPSetup},  \mathsf{PartyGen}, \mathsf{TPMembers}, \consensus)$ satisfying the properties outlined in Section \ref{prel:consensus:fundam-properties}.

\begin{enumerate}
	\item $\configtxgen (1^{\secpar},\ladmlist, \ordererlist, \peerlist, \policy)$
	lets the $\msp$ to initialize the $\sysname$ system. The initialization is optionally based on predetermined initial system participants, where $\ladmlist,\ordererlist, \peerlist$ are lists containing public keys for Local administrators, orderers and peers respectively, as well as a preselected system policy $\policy$.
	\begin{enumerate}
		\item $\msp$ sets as $\sysparams$ the system parameters for the signature and the hash function,
		as well as the consensus algorithm by running $\tpsetup$.
		\item Computes a random key pair $\signgen{\allowbreak\publicmsp{}}{\allowbreak\secretmsp{}}$.
		\item Assembles and outputs the genesis block $\block{0}$ (serving as the initial configuration block) by copying $\publicmsp{}, \sysparams$ and $[\ladmlist, \ordererlist, \peerlist, \policy]$ from the algorithm inputs.
		\item Initializes empty lists in $\msp$ memory \\ $[\ladmlist_{\msp}, \ordererlist_{\msp}, \peerlist_{\msp},\policy,\mspopers]$ where $\mspopers$ denotes a pending revoke operation list.
		\item Copies $\policy$ to $\policy_{\msp}$.
	\end{enumerate}
	The genesis block $\block{0}$ (as the blockchain $\blockchain$ in general) is public, while the rest of the outputs remain private to  $\msp$. For the following algorithms and protocols we assume that the security parameter and the system parameters are a default input.
	
	\item $\updateconfig ( \blockchain, \secretmsp{}, \nstate{\msp})$ enables $\msp$ to read system configuration information from its memory that is pending to be updated, and construct a new configuration block to make the new system configuration readable and valid in the blockchain by all system participants.
	\begin{enumerate}
		\item $\msp$ parses $\nstate{\msp}$ as $[\ladmlist_{\msp}, \ordererlist_{\msp}, \peerlist_{\msp}, \policy_{\msp}]$.
		\item Assembles a configuration update transaction 
		$\tx_{u} = \allowbreak[\ladmlist_{\msp},\allowbreak \ordererlist_{\msp}, \peerlist_{\msp}, \policy_{\msp}]$.
		\item Parses $\readconfig (\blockchain)$ as $\peerlist_{\blockchain}$.
		\item Sends signed transaction $\sigma_{\msp}(\tx_{u})$ to all $\aggr{i}{} \in \peerlist_{\blockchain}$.
	\end{enumerate}
	Since $\policy$ does not apply to transactions signed by $\msp$, the configuration update transaction is promptly appended to $\blockchain$ by all aggregators, resulting in  public output $\blockchain'$.
	
	\item $\updatepolicy ( \nstate{\msp},\policy)$ enables $\msp$ to update system policy parameters. 
	On input of a new system policy $\policy$, $\msp$ copies it to $\nstate{\msp}[\policy]$, overwriting the previous policy. The algorithm outputs the new updated $\nstate{\msp}'$.
	
	\item $\readconfig (\blockchain)$ can be run by any system participant to recover the current system configuration.
	\begin{enumerate}
		\item Parses $\blockchain$ as a series of blocks $\block{i}$.
		\item From the set of blocks marked as ``configuration" blocks where $\block{i}[type="C"]$, selects the block $\block{c}$ with the greatest height $c$.
		\item  Parses and outputs $\block{c}$ as $([\ladmlist_{\blockchain}, \ordererlist_{\blockchain}, \peerlist_{\blockchain}],\policy_{\blockchain})$.
	\end{enumerate} 
	
	\item $\orderersetup ()$ is  run by an orderer $\orderer{i}$ initializing its credentials and state. It computes and outputs signing keys as $\signgen{\publickey{\orderer{i}}{}}{\secretkey{\orderer{i}}{}}$ and initializes an signed transaction list in memory $\nstate{\orderer{i}}[\orderertxlist]$.
	
	\item $\ordereradd \{ \orderer{i}(\publickey{\orderer{i}}{}, \secretkey{\orderer{i}}{}) \leftrightarrow \\ \msp(\publicmsp{},\secretmsp{},\nstate{\msp}[\ordererlist_{\msp}],\blockchain) \} $ is an interactive protocol between an orderer $\orderer{i}$ and the system $\msp$ in order to add that orderer in the system:
	\begin{enumerate}
		\item $\orderer{i}{}$ first creates a physical identity proof $\pi$, then submits $\pi$ and $\publickey{\orderer{i}}{}$ to $\msp$.
		\item $\msp$ verifies $\pi$. Then it parses $(\ordererlist_{\blockchain})$ from the output of $\readconfig (\blockchain)$. Check that $(\publickey{\orderer{i}}{} \notin \ordererlist_{\msp}) \land (\publickey{\orderer{i}}{} \notin \ordererlist_{\blockchain})$. If all verifications hold, add $\publickey{\orderer{i}}{}$ to its local orderer list $\ordererlist_{\msp}$ and return ``1" to $\orderer{i}{}$, else return ``0" with an error code.
	\end{enumerate}
	
	\item $\localadmsetup () $ is an algorithm run by a $\localadmin{}$ to initialize its credentials and state and create a new device group $\iotGroup{}$.
	A Local Administrator computes and outputs signing keys as $\signgen{\publickey{\localadmin{i}}{}}{\secretkey{\localadmin{i}}{}}$. 
	Allocates memory for storing group aggregators' and sensors' public keys as $\nstate{\localadmin{i}}[ \agglistladmin,\senslistladmin]$.
	
	\item $\localadmreg \{ \localadmin{i} (\publickey{\localadmin{}}{}, \secretkey{\localadmin{}}{}) \leftrightarrow \\
	\msp(\publicmsp{},\secretmsp{}, \nstate{\msp}[\ladmlist_{\msp}],\blockchain)  \}$ is an interactive protocol between a local group administrator $\localadmin{i}$ and  $\msp$ in order to add $\localadmin{i}$ in the system.
	\begin{enumerate}
		\item $\localadmin{i}$ creates a physical identity proof $\pi$, then submits $\pi$ and $\publickey{\localadmin{}}{}$ to $\msp$.
		\item $\msp$ verifies $\pi$. Then it parses $(\ladmlist_{\blockchain})$ from the output of $\readconfig (\blockchain)$. Check that $(\publickey{\localadmin{i}}{} \notin \ladmlist_{\blockchain}) \land (\publickey{\localadmin{i}}{} \notin \ladmlist_{\msp})$. If all verifications hold, add $\publickey{\localadmin{i}}{}$ to $\ladmlist_{\msp}$ in $\nstate{\msp}$ and return ``1" to $\localadmin{i}{}$, else return ``0" with an error code.
	\end{enumerate}
	
	\item $\aggsetup \{\localadmin{i} (\publickey{\localadmin{i}}{}, \secretkey{\localadmin{i}}{},
	\nstate{\localadmin{i}}) \leftrightarrow \aggr{i}{j}() \}$ is an interactive protocol between an $\localadmin{i}$ and an aggregator $\aggr{i}{j}$ wishing to join group $\iotGroup{i}$.
	\begin{enumerate}
		\item $\aggr{i}{j}$ computes signing keys as $\signgen{\publicagg{i}{j}}{\secretagg{i}{j}}$ and initializes pending and write transaction sets $\pset{i}\rightarrow \emptyset, \txset{i}\rightarrow \emptyset$  in its $\nstate{\aggr{i}{j}}$.
		\item $\aggr{i}{j}$ creates a physical identity proof $\pi$, then submits $\pi$ and $\publicagg{i}{j}$ to $\localadmin{i}$.
		\item $\localadmin{i}$ verifies $(\publicagg{i}{j} \notin \agglistladmin)$ and $\pi$. If these verifications hold, it invokes $\aggAdd(\publicagg{i}{j})$ with $\msp$. If $\msp$ outputs ``1", it adds $\publicagg{i}{j}$ to $\agglistladmin$, sends an updated copy of $\agglistladmin$ to all $\aggr{i}{j} \in \agglistladmin$ and  $\publickey{\localadmin{i}}{}
		$ to $\aggr{i}{j}$. In all other cases it returns ``0".
		\item $\aggr{i}{j}$ copies $\publickey{\localadmin{i}}{}
		$ in its memory in $\nstate{\aggr{i}{j}}$.
	\end{enumerate}
	
	\item $\aggAdd \{\localadmin{i} (\publickey{\localadmin{i}}{}, \secretkey{\localadmin{i}}{}, \publicagg{i}{j}) \leftrightarrow \\ \msp(\publicmsp{},\secretmsp{},\nstate{\msp}[\peerlist_{\msp}],\blockchain) \}$ is an interactive protocol between a local administrator $\localadmin{i}$ wishing to add an aggregator to the system and $\msp$.
	\begin{enumerate}
		\item $\localadmin{i}$ computes $\sign{\secretkey{\localadmin{i}}{}}{\publicagg{i}{j}}{\sigma}$. Send $\sigma$ to $\msp$.
		\item $\msp$ computes $\svrfy{\publickey{\localadmin{i}}{}}{\publicagg{i}{j}}{\sigma}$. Checks that $(\publickey{\localadmin{i}}{} \in \ladmlist_{\msp}) \land b \land (\publicagg{i}{j} \notin \peerlist_{\msp}) ==  1$. If the verification holds, it parses $(\peerlist_{\blockchain})$ from the output of $\readconfig (\blockchain)$.  Check that $(\publicagg{i}{j} \notin \peerlist_{\blockchain})$. If the verification holds, add $\publicagg{i}{j}$ to $\peerlist_{\msp}$ and returns ``1" to $\localadmin{i}$. It returns ``0" in all other cases.
	\end{enumerate}
	
	\item $\aggupd \{\localadmin{i}(\secretkey{\localadmin{}}{},
	\publicsens{}{}) \leftrightarrow \aggr{i}{j}(
	\nstate{\aggr{i}{j}}[\senslistagg]) \}$ is an interactive protocol between $\localadmin{i}$ and an aggregator $\aggr{i}{j}$, both belonging to Group $i$. It is used when $\localadmin{i}$ wants to add a sensor public key $\publicsens{}{}$ to $\aggr{i}{j}$ and update its sensor list $\senslistagg$.
	\begin{enumerate}
		\item $\localadmin{i}$ computes $\sigma := \mathsf{Sign}(\secretkey{\localadmin{i}}{},\publicsens{}{})$. Send $\sigma$ to $\aggr{i}{j}$.
		\item $\aggr{i}{j}$ computes $\svrfy{\publickey{\localadmin{i}}{}}{\publicsens{}{}}{\sigma}$. Checks that $(\publicsens{}{} \notin \nstate{\aggr{i}{j}}) \land b == 1$. If the verification holds, it adds $\publicsens{}{}$ to $\senslistagg$\footnote{$\localadmin{i}$ should run the protocol with every aggregator in the group, however we present this with one aggregator for simplicity.} and returns ``1" to $\localadmin{i}$. It returns ``0" in all other cases.
	\end{enumerate}
	
	\item  $\sensjoin \{\localadmin{i}(
	\publickey{\localadmin{i}}{},\secretkey{\localadmin{i}}{},\nstate{\localadmin{i}}[\senslistladmin]) \leftrightarrow \sens{i}{j}(n)  \}$ is an interactive protocol between $\localadmin{i}$ of Group $i$ and a sensor $\sens{i}{j}$ wishing to join the system.
	\begin{enumerate}
		\item $\sens{i}{j}$ using the one-time signature scheme described in Section \ref{prel:onetime-sigs}:
		\begin{enumerate}
			\item Samples $k \leftarrow (1^{\secpar})$.
			and stores it in $\nstate{\sens{i}{j}}$.
			\item Runs $\otkeygen{\publicsens{i}{j}}{\secretsens{i}{j}}{\ell = 1}{n}$\footnote{This step is typically computed by a powerful device.}
			\item Stores $\ell = 1$ to $\nstate{\sens{i}{j}}$ where $\ell$ denotes the current ``index" in the hash chain.
			\item Creates a physical identity proof $\pi$ 
			\item Sends $(\pi,\publicsens{i}{j})$ to $\localadmin{i}$
		\end{enumerate}
		\item $\localadmin{i}$ checks $\mathsf{Vrfy}(\pi) \land (\publicsens{i}{j} \notin \senslistladmin ) \\ \land \aggupd(\secretkey{\localadmin{i}}{},
		\publicsens{i}{j}) ==1$ $\forall \aggr{i}{j} \in \iotGroup{i}$. $\localadmin{i}$ 
		adds $\publicsens{i}{j}$ to $\senslistladmin$, else it outputs ``0".
	\end{enumerate}	
	
	\item  $\senssenddata \{\sens{i}{j}(\publicsens{i}{j},\secretsens{i}{j},
	m,\nstate{\sens{i}{j}}) \leftrightarrow \\ \aggrset{x}(
	\nstate{\aggr{}{}}[\senslistagg,\siglistagg,\txlist]) \}$ is a protocol between sensor $\sens{i}{j}$ broadcasting data and a subset of aggregators $\aggrset{x} \subseteq \{\aggrset{i}\}$ (where $\{\aggrset{i}\}$ is the aggregator set in $\iotGroup{i}$).
	\begin{enumerate}
		\item For sending data $m$, $\sens{i}{j}$ computes  $\otsign{\secretsens{}{}}{\secretsens{}{}}{\nstate{\sens{i}{j}}'}{\nstate{\sens{i}{j}}}{m}{\sigma}$
		
		\item $\sens{i}{j}$ broadcasts $\sigma$ to $\aggrset{x}$.
		\item $\aggr{k}{}$ runs $\otverify{\publicsens{i}{j}}{m}{\sigma}$.\footnote{To avoid redundancy, the protocol can be improved by deterministically defining a ``responsible" aggregator for each transaction as discussed previously in this section.} If $b == 1$ it runs $\aggagree$ with all other aggregators in the group. If no ``alarm'' message $m^{A}$ from some other aggregator is received within some time $\delta$, it adds $m,\sigma$ to $\pset{i}$. If at least one ``alarm'' message is received, it outputs $\bot$.
	\end{enumerate}
	
	\item $\aggagree \{ \aggr{k}{}(\secretagg{k}{},\agglistladmin,\publicsens{i}{j},m, \sigma) \leftrightarrow \aggrset{}([\secretagg{i}{},\publicsens{i}{j},m'])\}$ is a protocol between an aggregator in $\iotGroup{i}$ and all other aggregators in the group. The purpose of this protocol is to detect any MITM attacks, and verifies that no aggregator in the group has received any message $m',\sigma'$ from $\publicsens{i}{j}$ where $m \neq m'$\footnote{This protocol does not require that all other aggregators in the group are reachable, therefore it does not require a reply from all aggregators to complete.}.
	\begin{enumerate}
		\item $\aggr{k}{}$ for payload $\mu = (\publicsens{i}{j},m, \sigma)$ computes $s = \Sign(\secretagg{k}{},\mu)$ using an EU-CMA signature scheme and sends $s,\mu$ to all $\aggr{i}{} \in \aggrset{}$.
		\item Each $\aggr{i}{}$ checks if it received a message $m'$ with signature $\sigma$ from sensor $\publicsens{i}{j}$ where $m \neq m'$. If there's no such message, it outputs $\bot$. Else it sends an ``alarm'' message $m^{A}$ and respective signature $s$ to $\aggr{k}{}$ and keeps a record in its log.
	\end{enumerate}
	
	\item $\aggsendtx \{\aggrset{}([\publicagg{i}{},\secretagg{i}{},\nstate{\aggr{i}{}},\blockchain]) \leftrightarrow \ordererset(\publickey{\orderer{j}}{}, \secretkey{\orderer{j}}{},\nstate{j})  \}$ is an interactive protocol between all aggregators $\aggr{i}{} \in \aggrset{}$ and all orderers $\orderer{j} \in \ordererset$. It is initiated when an aggregator wishes to submit a transaction for validation in the system and eventually store it in the blockchain.
	\begin{enumerate}
		\item An $\aggr{i}{} \in \aggrset{}$ parses $\pset{i}$ in $\nstate{\aggr{i}{}}$ as a set of transactions $\{\tx \}$.
		\item $\aggr{i}{}$ samples a nonce $n \leftarrow (1^{\secpar})$ and appends it to $\{\tx \}$.
		\item $\aggr{i}{}$ computes $\sign{\secretagg{i}{}}{\{\tx \}}{\sigma}$. Send $\sigma$ to all other $\aggr{j}{} \in \aggrset{x}$.
		\item Each $\aggr{j}{}$, parses $(\peerlist_{\blockchain})$ from the output of $\readconfig (\blockchain)$. Computes $\svrfy{\publicagg{i}{}}{\{\tx \}}{\sigma}$. If $(\publicagg{i}{} \in \peerlist_{\blockchain}) \land b==1$, compute\\ $\sign{\secretagg{j}{}}{\{\tx \}}{\sigma_{j}}$. Send $\sigma_{j}$ to $\aggr{i}{}$.
		\item $\aggr{i}{}$ parses $\readconfig(\blockchain) \rightarrow \policy_{\blockchain} \rightarrow \tau$ where $\tau$ the minimum required number of signatures for a transaction to be submitted on the blockchain, as defined by policy $\policy$. 
		\item If $| \{ \sigma_{j} \} | > \tau$, select a reachable orderer $\orderer{}$, send $\{ \sigma_{j} \}$, copy $\{\tx \} \rightarrow \txset{}$ and set $\pset{} \rightarrow \emptyset$.
		\item The orderer $\orderer{}$ parses $\nstate{j} \rightarrow \orderertxlist$,\\ $\{\tx \} \rightarrow n$, $\readconfig(\blockchain) \rightarrow \peerlist_{\blockchain}, \policy_{\blockchain}$ and $\policy_{\blockchain} \rightarrow \tau$ then checks:
		\begin{enumerate}
			\item $| \{ \sigma_{j} \} | > \tau$
			and $n \notin \orderertxlist$
			\item Compute $ \svrfy{\publicagg{j}{}}{\{\tx \}}{\sigma_{j}}, \forall j$ then \\ $\prod b_{j} ==1$
			\item $\prod_{j}^{}(\{\publicagg{j}{}\}  \in \peerlist_{\blockchain}) ==1$
		\end{enumerate} 
		If the checks are valid, stores $| \{ \sigma_{j} \} |$ in its $\nstate{\orderer{i}}$ and replies ``1" to $\aggr{i}{}$ as a confirmation, else it replies ``0".
		\item If $\orderer{i}$ has created a new block containing ordered transactions, it runs $\consensus$ to update the blockchain. 
		\item If $\consensus$ succeeds, it runs $\updateBC$ with all aggregators to update to the new $\blockchain'$.
	\end{enumerate}
	
	\item $\consensus ([[\publickey{\orderer{j}}{}]_{j=1}^{n}, \secretkey{\orderer{i}}{},\nstate{i},\blockchain]_{i=1}^{n})  := \blockchain'$\\
	The exact protocol functionality is described in the system parameters $\sysparams$\footnote{This is equivalent to Hyperledger's ``pluggable" consensus, which is defined in the genesis block.} and follows the definition provided in Section \ref{prel:consensus}. In general this protocol is executed among all orderers $\orderer{i} \in \ordererset$ where they agree on a new updated blockchain $\blockchain'$. In Appendix \ref{constr-pluggedconsens} we provide a concrete instantiation of a consensus algorithm for our construction.
	
	\item  $\updateBC \{\orderer{i}(\publickey{\orderer{i}}{}, \secretkey{\orderer{i}}{} ,\nstate{\orderer{i}},\blockchain')\leftrightarrow\\
	\aggrset{}([\publicagg{x}{},\secretagg{x}{},\nstate{\aggr{x}{}},\blockchain]) \}$ is initiated by an orderer $\orderer{i}$ to append a new block in the blockchain.
	\begin{enumerate}
		\item $\orderer{i}$ parses its $\nstate{\orderer{i}}$ to retrieve the agreed blockchain update signature set $\{\sigma_{x}' \}$
		\item $\orderer{i}$	computes $\sign{\secretkey{\orderer{i}}{}}{(\block{i},\{\sigma_{x}' \})}{\sigma}$ where $\blockchain' := \blockchain||\block{i}$ and sends $\sigma$ to all $\aggr{x}{} \in \aggrset{}$.
		\item Each $\aggr{x}{}$ computes $\svrfy{\publickey{\orderer{i}}{}}{\block{i}||\{\sigma_{x}' \}}{\sigma}$ and checks if $b==1$. Then it parses $\sigma$ as a transaction set $\{\tx\}$ and removes these from $\txset{x}$. Then it updates $\blockchain$ to $\blockchain'$, else it outputs $\bot$.
	\end{enumerate}

	\item $\noderem(\publickey{i}{},\sigma,\nstate{\msp}, \blockchain) $ is initiated by $\msp$ to revoke credentials of any system participant.
	
	$\msp$ parses $\readconfig(\blockchain)$ as $[\ladmlist_{\blockchain}, \ordererlist_{\blockchain}, \peerlist_{\blockchain}]$. It verifies $\sigma$ (if the remove operation was initiated by a $\localadmin{}$) and checks if $\publickey{i}{}$ exists in $[\ladmlist_{\blockchain}, \ordererlist_{\blockchain}, \peerlist_{\blockchain}]$ or in its $[\ladmlist_{\msp}, \ordererlist_{\msp}, \peerlist_{\msp}] \in \nstate{\msp}$ in case participation privileges for $\publickey{i}{}$ have not yet been updated on the $\blockchain$ through $\updateconfig$. If it finds a match in the blockchain lists, it creates a remove operation $R := (\publickey{i}{}, ``rm")$ and adds $R$ to $\mspopers$, else if it finds a match in its state lists it removes it from the respective list, else it outputs $\bot$. If $\publickey{i}{} \in \peerlist_{\blockchain} \lor \publickey{i}{} \in \peerlist_{\msp}$, it also informs $\localadmin{i}$.

	\item $\grprem (\publickey{i}{},\nstate{\localadmin{}}[ \agglistladmin,\senslistladmin])$ is initiated by a Local Administrator to revoke credentials of an aggregator or sensor in its group. $\localadmin{}$ checks if $\publickey{i}{} \in [ \agglistladmin,\senslistladmin]$. If it finds a match and $\publickey{i}{} \in \agglistladmin$, it $\localadmin{i}$ computes $\sign{\secretkey{\localadmin{i}}{}}{\publicagg{i}{},"R"}{\sigma}$. Sends $(\sigma,\publicagg{i}{},"R")$ to $\msp$. On receiving successful removal from $\msp$ (after it invokes $\noderem$), $\localadmin{}$ removes $\publickey{i}{}$ from $\agglistladmin$. If $\publickey{i}{} \in \senslistladmin$, it invokes $\aggremsens$ with all $\aggr{}{} \in \agglistladmin$. After successful completion, it removes $\publickey{i}{}$ from $\senslistladmin$. 
	
	\item $\aggremsens \{\localadmin{i}(\secretkey{\localadmin{}}{},
	\publicsens{}{}) \leftrightarrow \aggr{i}{j}(
	\nstate{\aggr{i}{j}}[\senslistagg]) \}$ is initiated by a Local Administrator as a subroutine of $\grprem$ to revoke credentials of a sensor in its group.
	\begin{enumerate}
		\item $\localadmin{i}$ computes $\sigma := \mathsf{Sign}(\secretkey{\localadmin{i}}{},\publicsens{}{}``R")$. Send $\sigma$ to $\aggr{i}{j}$.
		\item $\aggr{i}{j}$ computes $\svrfy{\publickey{\localadmin{i}}{}}{\publicsens{}{}}{\sigma}$. Checks that $(\publicsens{}{} \notin \nstate{\aggr{i}{j}}) \land b == 1$. If the verification holds, it removes $\publicsens{}{}$ from $\senslistagg$ and returns ``1" to $\localadmin{i}$. It returns ``0" in all other cases.
	\end{enumerate}

\end{enumerate}

\section{Evaluation details}
\label{apdx:evaldetails}
Algorithms \ref{measurements-sensor-pseudocode} and \ref{measurements-aggregator-pseudocode} show the pseudocode for our evaluations on the sensor and aggregator side respectively. We denote by timer1 the signature computation time, by timer2 the verification time and by timer3 the total verification time, as previously shown in Table \ref{measurements-table}.

\else
\begin{acks}
	Foteini Baldimtsi and Panagiotis Chatzigiannis were supported by NSA 204761.
\end{acks}
\fi

\end{document}